\documentclass[acmsmall,screen]{acmart}

\usepackage[utf8]{inputenc}
\usepackage{amsfonts}
\usepackage{amsthm}
\usepackage{amsmath}
\usepackage{wrapfig}
\usepackage{cleveref}
\usepackage{subcaption}
\usepackage{listings}
\usepackage{xspace}
\usepackage{pifont}      
\usepackage{color}       
\usepackage{algorithm}
\usepackage[noend]{algpseudocode}
\usepackage{algorithmicx}

\newtheorem{theorem}{Theorem}
\theoremstyle{definition}
\newtheorem{proposition}{Proposition}
\newtheorem{definition}{Definition}

\newcommand{\project}{WaveCert\xspace}

\newcommand{\allconfig}{\ensuremath{\mathcal{C}}}
\newcommand{\initialconfig}{\ensuremath{\mathcal{I}}}
\newcommand{\channel}{\ensuremath{\mathrm{Channel}}}
\newcommand{\config}{c}
\newcommand{\cutpoint}{P}

\newcommand{\df}{\ensuremath{\mathrm{df}}}
\newcommand{\llvm}{\ensuremath{\mathrm{llvm}}}
\newcommand{\seq}{\ensuremath{\mathrm{seq}}}
\newcommand{\mem}{\ensuremath{\mathsf{Mem}}}

\newcommand{\hint}{\ensuremath{\mathit{hint}}}

\newcommand{\canonsched}{\mathsf{Canon}}
\newcommand{\seqlang}{Seq\xspace}
\newcommand{\tracesdf}{\mathsf{Traces}_\df}
\newcommand{\tracesseq}{\mathsf{Traces}_\seq}

\newcommand{\varlist}[1]{\mathbf{#1}\xspace}
\newcommand{\term}{\mathsf{Terms}}

\newcommand{\readperm}{\ensuremath{\texttt{read}}}
\newcommand{\writeperm}{\ensuremath{\texttt{write}}}
\newcommand{\permleq}{\ensuremath{\leq}}
\newcommand{\permle}{\ensuremath{<}}
\newcommand{\permdisjoint}{\ensuremath{\mathit{disjoint}}}

\newcommand{\numevalprograms}{21\xspace}
\newcommand{\numbugs}{8\xspace}

\algnewcommand{\LineComment}[1]{\State \(\triangleright\) #1}

\lstdefinelanguage{llvm}{
  morecomment = [l]{;},
  morestring=[b]", 
  sensitive = true,
  classoffset=0,
  morekeywords={
    define, declare, global, constant,
    internal, external, private,
    linkonce, linkonce_odr, weak, weak_odr, appending,
    common, extern_weak,
    thread_local, dllimport, dllexport,
    hidden, protected, default,
    except, deplibs,
    volatile, fastcc, coldcc, cc, ccc,
    x86_stdcallcc, x86_fastcallcc,
    ptx_kernel, ptx_device,
    signext, zeroext, inreg, sret, nounwind, noreturn,
    nocapture, byval, nest, readnone, readonly, noalias, uwtable,
    inlinehint, noinline, alwaysinline, optsize, ssp, sspreq,
    noredzone, noimplicitfloat, naked, alignstack,
    module, asm, align, tail, to,
    addrspace, section, alias, sideeffect, c, gc,
    target, datalayout, triple,
    blockaddress
  },
  morekeywords={
    add, fadd, sub, fsub, mul, fmul,
    sdiv, udiv, fdiv, srem, urem, frem,
    and, or, xor,
    icmp, fcmp,
    eq, ne, ugt, uge, ult, ule, sgt, sge, slt, sle,
    oeq, ogt, oge, olt, ole, one, ord, ueq, ugt, uge,
    ult, ule, une, uno,
    nuw, nsw, exact, inbounds,
    phi, call, select, shl, lshr, ashr, va_arg,
    trunc, zext, sext,
    fptrunc, fpext, fptoui, fptosi, uitofp, sitofp,
    ptrtoint, inttoptr, bitcast,
    ret, br, indirectbr, switch, invoke, unwind, unreachable,
    malloc, alloca, free, load, store, getelementptr,
    extractelement, insertelement, shufflevector,
    extractvalue, insertvalue,
  },
  alsoletter={\%},
  keywordsprefix={\%},
}
\renewcommand\footnotetextcopyrightpermission[1]{}
\setcopyright{none}

\makeatletter
\let\@authorsaddresses\@empty
\makeatother

\begin{document}

\title{WaveCert: Translation Validation for Asynchronous Dataflow Programs via Dynamic Fractional Permissions}

\author{Zhengyao Lin}
\affiliation{
  \department{Computer Science Department}
  \institution{Carnegie Mellon University}
  \country{USA}
}
\email{zhengyal@cmu.edu}

\author{Joshua Gancher}
\affiliation{
  \department{Computer Science Department}
  \institution{Carnegie Mellon University}
  \country{USA}
}
\email{jgancher@andrew.cmu.edu}

\author{Bryan Parno}
\affiliation{
  \department{Computer Science Department}
  \institution{Carnegie Mellon University}
  \country{USA}
}
\email{parno@cmu.edu}

\begin{abstract}
Coarse-grained reconfigurable arrays (CGRAs) have gained attention in recent years due to their 
promising power efficiency compared to traditional von Neumann architectures.
To program these architectures using ordinary languages such as C, 
a dataflow compiler must transform the original sequential, imperative program
into an equivalent dataflow graph, composed of dataflow operators running
in parallel.
This transformation is challenging since the asynchronous nature of dataflow
graphs allows out-of-order execution of operators, leading to behaviors not present in the original imperative programs.

We address this challenge by developing a translation validation technique for dataflow compilers
to ensure that the dataflow program has the same behavior as the original imperative program 
on all possible inputs and schedules of execution.
We apply this method to a state-of-the-art dataflow compiler targeting the RipTide CGRA architecture.
Our tool uncovers 8 compiler bugs where the compiler outputs incorrect dataflow
graphs, including a data race that is otherwise hard to discover via testing.
After repairing these bugs, our tool verifies the correct compilation of all programs in the RipTide benchmark suite.

\end{abstract}

\maketitle 

\section{Introduction}



Even as Moore's Law ends, many applications still demand better power efficiency, 
e.g., for high-performance machine learning
or for ultra-low-power environmental monitoring.
A flood of new hardware architectures has emerged to meet this demand,
but one attractive approach that retains general-purpose programmability
is a class of dataflow architectures called \emph{coarse-grained reconfigurable arrays} (CGRAs)~\cite{cgra-survey}.
In a CGRA architecture, an array of processing elements (PEs) is physically laid out on the chip.
Each PE can be configured to function as a high-level \emph{operator} that performs an arithmetic operation, a memory operation, or a control-flow operation.
These PEs communicate via an on-chip network, reducing data movement costs compared to traditional von Neumann architectures, thus making a CGRA more power efficient.

Abstractly, CGRAs run \emph{dataflow programs}, or networks of 
operators that execute independently and asynchronously. While CGRAs are able to 
run dataflow programs efficiently due to their inherent parallelism, 
this parallelism also introduces new, subtle correctness issues due to the potential for data races. As a result, dataflow programs are not written directly,
but are instead \emph{compiled} from sequential, imperative 
languages. However, the gap between imperative and dataflow programming
means that these compilers are complex; to maximize performance,
they must use a variety of program analyses to convert imperative programs
into equivalent distributed ones. 

As CGRA architectures continue to emerge and rapidly evolve,
we argue that \emph{formal verification} --- mathematically
sound proofs of system correctness --- 
should be integrated early on; 
past experience teaches us that belatedly retrofitting them may be costly or simply infeasible.
Since CGRAs cannot execute as intended without correct dataflow programs, 
we begin this effort by verifying the translation from imperative programs to dataflow programs.

In this work, we propose using \emph{translation validation}~\cite{pnueli-tv} to
certify the results of dataflow compilers. 
In our case, translation validation amounts to checking that 
the output dataflow program from the compiler has the same behaviors as the
input imperative program.
Compared to traditional testing, translation validation provides a much stronger guarantee, as it checks that the imperative and the dataflow programs behave the same on all possible inputs.
Additionally, in contrast to full compiler verification, 
translation validation allows the use of unverified compilation toolchains,
which is crucial for the fast-developing area of CGRA architectures.

Our main challenge for leveraging translation validation is managing the 
asynchrony of dataflow programs. Most 
translation validation works~\cite{pnueli-tv,necula-tv,kundu-tv,data-driven-eq}
prove equivalence between input and output programs via 
\emph{simulation}, or relational invariants between the states of the two
programs. Since dataflow programs have exponentially more states than
imperative programs (arising from asynchronous scheduling decisions), constructing this 
simulation relation directly would be quite tricky. 

Instead, we utilize a two-phase approach. First, we 
use symbolic execution to prove that there is a
simulation between the input imperative program and the output dataflow program
\emph{on a canonical schedule}. In the canonical schedule, the dataflow
operators are scheduled to fire in a similar order to their counterparts in the 
original imperative program. 

Next, we prove that any possible schedule of the dataflow program must
converge to the same final state as the canonical schedule; i.e., it is 
\emph{confluent}~\cite{term-rewriting}.
For our setting, proving confluence requires showing that the dataflow program
does not contain any data races. 
Inspired by fractional permissions~\cite{fractional-perm}, we augment our symbolic execution with
\emph{linear permission tokens} to track ownership of memory locations.
In contrast with interactive program logics (e.g., concurrent separation
logic~\cite{concurrent-separation-logic}) that typically require manual annotations,
we automatically compute \emph{dynamic} permissions that 
may flow in arbitrary (sound) ways.

Putting together both phases, we prove that the
imperative program is equivalent to the dataflow program on all possible inputs
and schedules.
Furthermore, equivalence with the sequential program also implies liveness and
deadlock freedom for the dataflow program, as it enforces that the dataflow
program can make progress whenever the sequential program can.


We have implemented this technique in \project, a translation validation system
targeting the state-of-the-art CGRA architecture RipTide~\cite{riptide},
which operates via a compiler from LLVM~\cite{llvm} programs to dataflow programs.
Using our tool, we found \numbugs compiler bugs where the RipTide dataflow compiler generated incorrect dataflow programs.
All of these bugs were confirmed by the developers of the RipTide compiler.
One of these bugs allows the compiler to emit a dataflow program with data races, which are hard to discover via testing and costly to fix after deployment.

To summarize, our contributions in this work are:
\begin{itemize}
    \item A novel two-phase translation validation technique to prove that the output dataflow program has equivalent behavior as the input imperative program, capturing both correctness and liveness properties.
    \item An implementation of our technique, \project\footnote{Source code is available in an GitHub repository~\cite{anonymous-repo}}, targeting the RipTide dataflow compiler.
    \item An evaluation of \project on the RipTide dataflow compiler, which reveals \numbugs compilation bugs. Most verification result takes around 10 seconds with a maximum of about 30 seconds.
\end{itemize}



\begin{figure}
    \centering
    \includegraphics[width=0.55\linewidth]{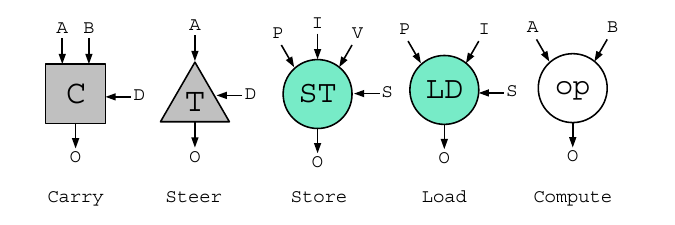}
    \Description{Examples of operators.}
    \caption{Examples of operators. \texttt{Carry} and \texttt{Steer} manage
    control flow, \texttt{Store} and \texttt{Load} manipulate the main memory,
    and \texttt{op} $\in \{+, *, <, \dots\}$ executes pure arithmetic operations.}
    \label{fig:operators}
\end{figure}

\section{Preliminaries: Dataflow Programs, CGRAs, and RipTide}
\label{sec:prelim}

Dataflow programming~\cite{dataflow-programming-survey} is an old idea dating
back to the 1960s and 1970s.
Early foundational works include the dataflow architecture by Dennis et al.~\cite{dennis-dataflow}, as well as theoretical analysis of models of dataflow programs such as computation graphs~\cite{computation-graphs} and Kahn process networks~\cite{kahn-process-networks}.
The main motivation in these early works is to achieve a greater degree of instruction-level parallelism~\cite{ackerman-dataflow}.

In the intervening years, the community has explored a wide variety of dataflow architectures~\cite{dataflow-architectures,yazdanpanah-dataflow}.
Among these dataflow architectures, \emph{coarse-grained reconfigurable arrays} (CGRAs)~\cite{cgra-survey} stand out due to their promising power efficiency: recent work~\cite{snafu,riptide} is able to support general programmability, while reducing the energy cost by orders of magnitude compared to a von Neumann core. 

To utilize the efficiency and parallelism in CGRAs, the predominant general-purpose way to program CGRAs is to compile a sequential, imperative program into a dataflow graph via a dataflow compiler, as seen in CGRA projects such as PipeRench~\cite{piperench}, WaveScalar~\cite{wavescalar}, and more recently, RipTide~\cite{riptide} and Pipestitch~\cite{pipestitch}.

In this work, we focus on the approach taken by RipTide~\cite{riptide}, which has a dataflow compiler from arbitrary LLVM functions to dataflow programs.
We briefly introduce LLVM and dataflow programs below.

\paragraph*{LLVM IR}
LLVM~\cite{llvm} is a compiler framework and features an intermediate representation called LLVM IR.
Compiler writers can translate a higher-level language such as C or Rust into LLVM IR, and then the LLVM framework can perform optimizations on LLVM IR and produce machine code in various target architectures.
For an informal description of the semantics of the LLVM IR, we refer the reader to the official LLVM IR manual~\cite{llvm-ir}.
In this work, we use the subset of LLVM IR supported by the RipTide compiler, including basic control-flow, integer arithmetic, and memory operations.
Function calls and floating-point operations are not supported.

\paragraph*{Dataflow programs}
A \emph{dataflow program} is a Turing-complete representation of programs.
It is represented as a dataflow graph in which the nodes are called \emph{operators} and edges are called \emph{channels}.
Semantically, operators can be thought of as stateful processes running in
parallel, which communicate through channels, or FIFO queues of messages.
Repeatedly, each operator: waits to dequeue (a subset of) input channels to get
input values; performs a local computation,
optionally changing its local state or global memory; and pushes output values 
to (a subset of) output channels. 
When an iteration of this loop is done, we
say that the operator has executed or \emph{fired}.
Since different CGRAs have different strategies for scheduling
operators~\cite{cgra-survey}, we conservatively use asynchrony to model all
possible schedules. 

In \Cref{fig:operators}, we show the five most important types of operators in
RipTide~\cite{riptide}\footnote{\project supports all RipTide operators,
except for the \emph{Stream} operator, as described in \Cref{sec:limitations}.}.
In \Cref{fig:operators}, reading from left to right, we first have the control flow operators, 
\emph{carry} (\texttt{C}) and \emph{steer} (\texttt{T}).
The carry operator is used to support loop variables.
It has two local states: in the initial state, it waits for a value from \texttt{A}, outputs it to
\texttt{O}, and transitions to the loop state. In the loop state, it waits for
values from \texttt{B} and \texttt{D} (decider); if \texttt{D} is true, it
outputs \texttt{O} = \texttt{B}; otherwise it discards \texttt{B} and
transitions to the initial state.

The steer operator is used for conditional execution. It waits for values from
\texttt{A} and \texttt{D}. If \texttt{D} is true, it outputs \texttt{O} =
\texttt{A}; otherwise, it discards the value and outputs nothing.

To utilize the global memory, we have the \emph{load} (\texttt{LD}) and \emph{store}
(\texttt{ST}) operators. The load operator waits for values from \texttt{P}
(base), \texttt{I} (offset), and an optional signal \texttt{S} for memory
ordering; reads the memory location \texttt{P[I]}; and outputs the read
value to \texttt{O}. Similarly, the store operator waits for values from
\texttt{P} (base), \texttt{I} (offset), \texttt{V} (value), and an optional
signal \texttt{S} for memory ordering; stores \texttt{P[I]} = \texttt{V} in the
global memory; and optionally outputs a finish signal to \texttt{O}.

All other \emph{compute} operators (e.g., \texttt{op} $\in \{+, *, <\}$) wait for values from input channels,
perform the computation, and then output the result to output channels.

To compile an LLVM program to a dataflow program, the RipTide compiler first 
maps each LLVM instruction to its corresponding dataflow
operator, and then it enforces the expected control-flow and memory ordering semantics of the original program.
Using a control dependency analysis, it inserts steer operators to
selectively enable/disable operators depending on branch/loop conditions; 
Using a memory ordering analysis,
it inserts dataflow dependencies between load/store operators to prevent data races.

These analyses are quite complex and use various optimizations to allow more parallelism than the original LLVM program.
As shown in our evaluation in \Cref{sec:evaluation}, this process can easily have bugs.
Hence, in our work, we use translation validation to certify the correctness of dataflow compilation.





\section{Overview and an Example}
\label{sec:overview}



\project performs translation validation on the RipTide dataflow compiler, which compiles LLVM programs to dataflow programs.
Given the input program and the compiled output program (along with some hints
generated by the compiler), \project performs two checks: 
first, a \emph{simulation} check, which verifies that the LLVM program is
equivalent to the dataflow program on a restricted, canonical schedule of dataflow
operators; and second, a \emph{confluence} check, which proves that the choice of the canonical schedule for the dataflow program is general and all other schedules reach the same final state as the canonical schedule. 

\newcommand{\examplellvmname}{\texttt{@test}\xspace}
\newcommand{\examplellvmheader}{\texttt{\%header}\xspace}
\newcommand{\examplellvmbody}{\texttt{\%body}\xspace}
\newcommand{\examplellvmloopindex}{\texttt{\%i}\xspace}
\newcommand{\examplellvmparama}{\texttt{\%A}\xspace}
\newcommand{\examplellvmparamlen}{\texttt{\%len}\xspace}

\newcommand{\exampledfphii}{1\xspace}
\newcommand{\exampledfphilso}{2\xspace}
\newcommand{\exampledficmp}{4\xspace}
\newcommand{\exampledfst}{9\xspace}
\newcommand{\exampledfld}{7\xspace}
\newcommand{\exampledfinci}{6\xspace}
\newcommand{\exampledfincai}{8\xspace}
\newcommand{\exampledfsteeri}{3\xspace}
\newcommand{\exampledfsteerlso}{5\xspace}

\begin{figure}[ht!]
    \centering
    \begin{subfigure}[b]{0.5\linewidth}
        \centering
        \begin{lstlisting}[language=llvm,escapechar=|]
define void @test(i32* %A, i32* %B, i32 %len) {
entry:
  br label %header
header:
  %i = phi i32 [ 0, %entry ], [ %i_inc, %body ]|\label{line:example-phi-i}|
  %cond = icmp slt i32 %i, %len|\label{line:example-icmp}|
  br i1 %cond, label %body, label %end
body:
  %idx1 = getelementptr i32, i32* %A, i32 %i|\label{line:example-gep}|
  %A_i = load i32, i32* %idx1|\label{line:example-load}|
  %A_i_inc = add i32 %A_i, 1|\label{line:example-add-a-i}|
  %idx2 = getelementptr i32, i32* %B, i32 %i
  store i32 %A_i_inc, i32* %idx2|\label{line:example-store}|
  %i_inc = add i32 %i, 1|\label{line:example-add-i}|
  br label %header
end:
  ret void
}
        \end{lstlisting}
        \caption{Input LLVM function.}
        \label{fig:example-llvm}
    \end{subfigure}
    \begin{subfigure}[b]{0.4\linewidth}
        \centering
        \includegraphics[width=0.78\linewidth]{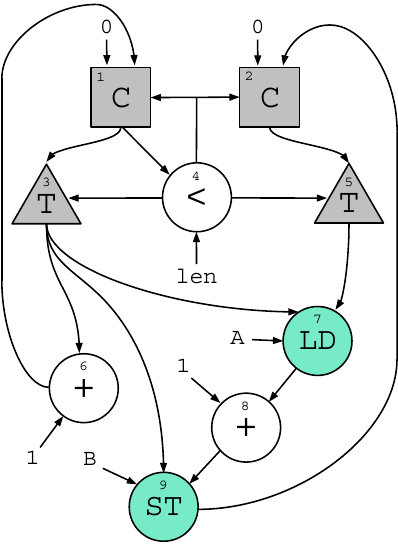}
        \caption{Output dataflow program.}
        \label{fig:example-dataflow}
    \end{subfigure}
    \caption{An example of input/output programs from the RipTide dataflow
    compiler. Most dataflow operators correspond to LLVM instructions: the
    comparison operator corresponds to line~\ref{line:example-icmp}, the add
    operators correspond to lines~\ref{line:example-add-a-i} and
    \ref{line:example-add-i}, while the load and store operators correspond to
    lines~\ref{line:example-load} and \ref{line:example-store}, respectively.}
    \Description{An example.}
    \label{fig:example}
\end{figure}


If a pair of input/output programs passes both checks, the two programs are equivalent (in the sense of \Cref{def:full-equivalence}).
Equivalence implies both \emph{safety} --- the dataflow program returns the
correct values --- and \emph{liveness} --- the dataflow program always
terminates correctly, and does not contain communication-related deadlocks.

\paragraph{Example}
In \Cref{fig:example}, we have an example input and output program from the dataflow compiler.
The input LLVM function \examplellvmname in \Cref{fig:example-llvm} contains a single loop with loop header block \examplellvmheader and loop body block \examplellvmbody.
The loop increments variable \examplellvmloopindex from \texttt{0} to \examplellvmparamlen, and at each iteration, updates \texttt{\%B[\%i] = \%A[\%i] + 1}.
%
%
The LLVM function \examplellvmname is compiled to the dataflow program in
\Cref{fig:example-dataflow}.

In addition to compute and memory operators, the compiler also inserts operators to faithfully implement the 
sequential semantics of the LLVM program.
The steer operators marked with
\texttt{T} enforce that the operators corresponding to the loop body only
execute when the loop condition is true.
The carry operator \exampledfphii is used for the loop induction variable \texttt{\%i} (line~\ref{line:example-phi-i}).
The carry operator \exampledfphilso is used for a loop variable inserted by the compiler to enforce a memory dependency:
since arrays \texttt{A} and \texttt{B}
could overlap, we have an additional dataflow path along operators $\exampledfst \to \exampledfphilso \to \exampledfsteerlso \to \exampledfld$
in order to enforce that the load in the next iteration waits until the previous store has finished.



\begin{figure*}
    \centering
    \includegraphics[width=\textwidth]{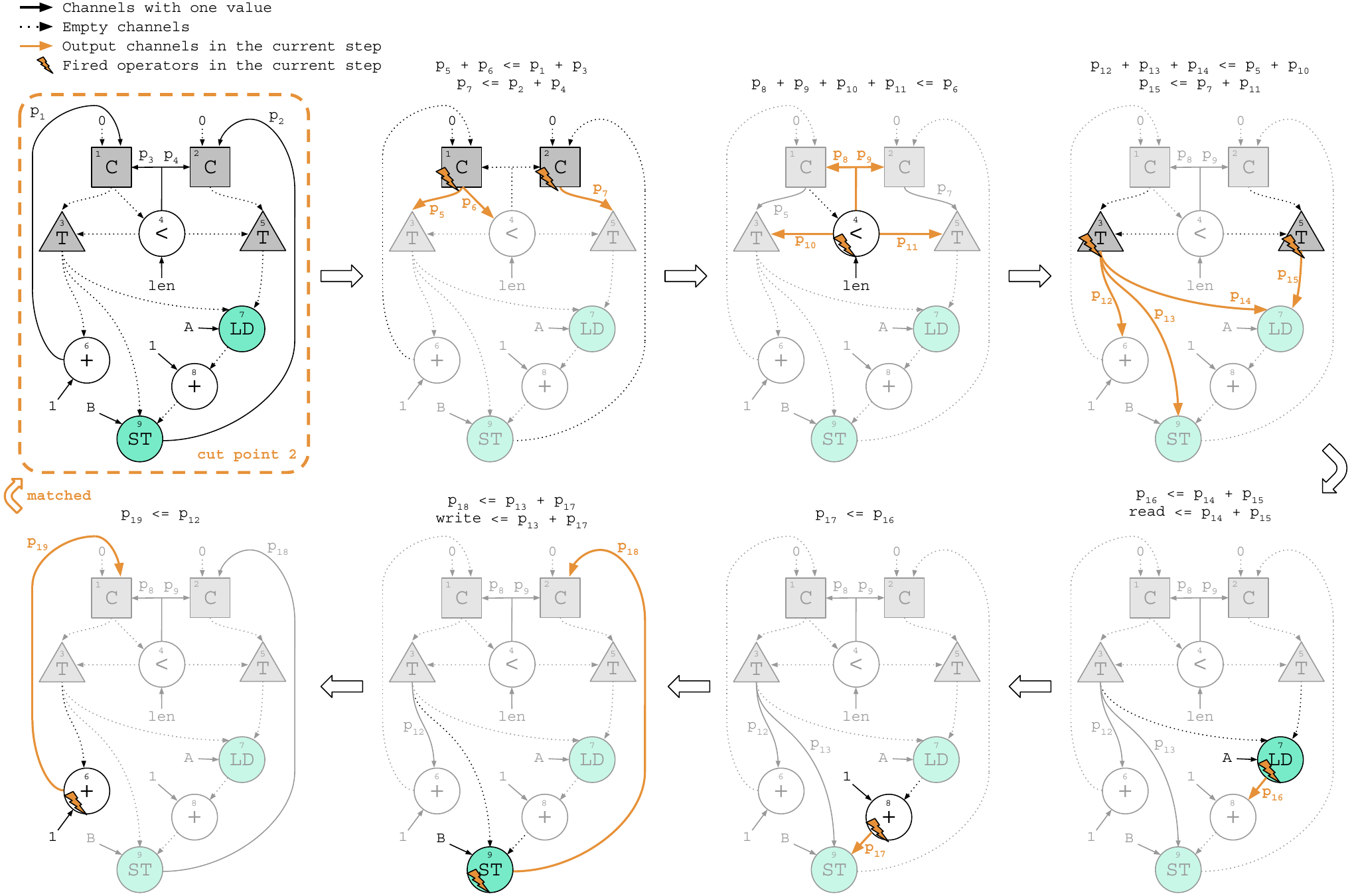}
    \caption{This plot shows a trace of execution from dataflow cut point 2 (upper left corner). The operators fired in this trace are \exampledfphii, \exampledfphilso, \exampledficmp, \exampledfsteeri, \exampledfsteerlso, \exampledfld, \exampledfincai, \exampledfst, \exampledfinci (which are inferred from the LLVM trace from cut point 2 to cut point 2). At the end of the trace, the configuration matches cut point 2 again. In each configuration, we mark the fired operators with orange lightning symbols. The emptiness of each channel after firing the operator(s) is represented by dashed vs solid lines. Before the firing of an operator, the input channels of the operator should be non-empty and marked with a solid line. After firing, the input channels are emptied, and the modified output channels will become non-empty and marked with a solid, orange line. The symbolic values in each channel are omitted, but their attached permission variables (starting with \texttt{p}) are marked. At the top of each configuration, we also indicate the permission constraints added at each step.}
    \Description{Cut point 2 trace.}
    \label{fig:example-dataflow-cut-point-2-to-cut-point-2}
\end{figure*}

\begin{figure}
    \centering
    \begin{tabular}{clll}
        Index & LLVM Program Point & Dataflow Configuration & Correspondence \\
        \hline
        1 & \texttt{\%entry} & Initial config & $\mem_\llvm = \mem_\df$ \\
        2 & \texttt{\%header} (from \texttt{\%body}) & First configuration in \Cref{fig:example-dataflow-cut-point-2-to-cut-point-2} & $\mem_\llvm = \mem_\df \wedge \texttt{\%i} = \texttt{i}$ \\
        $\bot$ & \texttt{\%end} & Final config & $\mem_\llvm = \mem_\df$
    \end{tabular}
    \caption{Cut points for the example in \Cref{fig:example}.
    For the second LLVM cut point, \project places it at the program point at the back edge from \texttt{\%body} to \texttt{\%header}.
    The correspondence equality $\mem_\llvm = \mem_\df$ means that the memory state should be the same at each cut point.
    The equality $\texttt{\%i} = \texttt{i}$ means that the LLVM variable $\texttt{\%i}$ at cut point 2 should be equal to the dataflow variable \texttt{i} referring to the value in the channel from operator \exampledfinci to \exampledfphii in \Cref{fig:example-dataflow-cut-point-2-to-cut-point-2}.
    We implicitly assume that all function parameters are equal (i.e. $\texttt{\%A} = \texttt{A} \wedge \texttt{\%B} = \texttt{B} \wedge \texttt{\%len} = \texttt{len}$).}
    \Description{Cut points for the example.}
    \label{fig:example-cut-points}
\end{figure}

\paragraph{Simulation Check}

The first step in \project's translation validation procedure is to check that the LLVM and dataflow program are equivalent on a canonical schedule of dataflow operators, where the dataflow operators are executed in the same order as their LLVM instruction counterparts.

This is done using a \emph{simulation relation}, or relational invariant, between the states of LLVM and dataflow programs.
\project constructs this relation by placing corresponding \emph{cut points} on both the LLVM and the dataflow sides.
A cut point symbolically describes a set of pairs of LLVM and dataflow
configurations satisfying a correspondence condition; a list of cut points together describes the simulation relation.
\Cref{fig:example-cut-points} shows the list of cut points for the example.
\project places cut point 1 for the initial configurations, cut point 2 for the loop structure, and a final cut point $\bot$ for all final configurations.
For each pair of LLVM and dataflow cut points, \project also infers a list of equations expressing the correspondence of symbolic variables in these cut points.

To automatically check that the proposed relation in \Cref{fig:example-cut-points} is indeed a simulation, we perform \emph{symbolic execution}~\cite{symbolic-execution} from all pairs of cut points (except for $\bot$) until they reach another pair of cut points.
On the LLVM side, we first symbolically execute cut point 1, which turns into two symbolic branches: reaching \texttt{\%header} after one loop iteration (cut point 2), or failing the loop condition and reaching \texttt{\%end} (cut point $\bot$).
For LLVM cut point 2, we also have two branches: one reaching cut point 2 again after a loop iteration, the other failing the loop condition and reaching cut point $\bot$.

For the dataflow side, we have a correspondence between a subset of LLVM instructions and a subset of dataflow operators, which is automatically generated by the dataflow compiler.
We use this mapping to infer the order in which we fire the dataflow operators, i.e., the canonical schedule.
For example, from LLVM cut point 2, there are two branches reaching cut point 2 and cut point $\bot$ respectively.
These two branches have two traces of LLVM instructions (excluding \texttt{br}):
\begin{itemize}
    \item For the first branch to cut point 2: lines \ref{line:example-phi-i}, \ref{line:example-icmp}, \ref{line:example-gep} - \ref{line:example-add-i};
    \item For the second branch to cut point $\bot$: lines \ref{line:example-phi-i}, \ref{line:example-icmp}.
\end{itemize}
For the first branch, the LLVM instructions map to dataflow operators \exampledfphii, \exampledfphilso, \exampledficmp, \exampledfld, \exampledfincai, \exampledfst, \exampledfinci (excluding ones that do not have a corresponding operator; operator \exampledfphilso corresponds to a \texttt{phi} instruction implicitly inserted during compilation for memory ordering).
LLVM instructions in the second branch map to dataflow operators \exampledfphii, \exampledfphilso, and \exampledficmp.
Therefore, for dataflow cut point 2, we execute these two traces of dataflow operators (and also any steer operator that can be executed), which gives us two dataflow configurations.
\Cref{fig:example-dataflow-cut-point-2-to-cut-point-2} shows the trace of execution of the first branch from cut point 2 to itself.
At the end of the execution, we can see that the configuration matches the cut point 2 we started with, except with a different memory state and symbolic variables (e.g. \texttt{i + 1} vs. \texttt{i} in the channel from operator \exampledfinci to \exampledfphii).

After a branch is matched to a cut point, we check the validity of the correspondence equations at the target cut points (given the assumptions made in the source cut point).
For example, for the branch shown in
\Cref{fig:example-dataflow-cut-point-2-to-cut-point-2}, we check that given the
source cut point correspondence ($\mem_\llvm = \mem_\df \wedge \texttt{\%i} =
\texttt{i}$), the target cut point correspondence holds:
\[
    \mem_\llvm[\texttt{\%B}[\texttt{\%i}] \mapsto \texttt{\%A}[\texttt{\%i}] + 1] = \mem_\df[\texttt{B}[\texttt{i}] \mapsto \texttt{A}[\texttt{i}] + 1] \wedge \texttt{\%i + 1} = \texttt{i + 1}
\]
If this check succeeds for all cut points, we soundly conclude that we have obtained a simulation relation between the LLVM program and the dataflow program on the canonical schedule.

\paragraph{Confluence Check}
\label{sec:example-confluence}

Once we have established a simulation relation between the LLVM program and a canonical schedule of the dataflow program, we check that the dataflow program is \emph{confluent}: all possible schedules converge to the canonical schedule in the final state (or synchronize with it infinitely often if the program is non-terminating).
This confluence result ensures that our simulation result for the canonical schedule is fully general for all other schedules.

We achieve this by using \emph{linear permission tokens}.
For this example, let us consider a simplified set of permission tokens $\{ 0, \readperm_1, \readperm_2, \writeperm \}$, where $0$ means no permission to read or write, and $\readperm$/$\writeperm$ means read/write permission to the entire memory, respectively.
Denote $p_1 + p_2$ as the disjoint sum of two permission tokens $p_1$ and $p_2$. The disjoint sum is a partial function satisfying $p + 0 = 0 + p = p$, $\readperm_1 + \readperm_2 = \writeperm$. However, $\writeperm + \writeperm$ is undefined.
Moreover, we partially order these permission tokens by $0 \permle \readperm_1, \readperm_2 \permle \writeperm$.
We will write $\readperm$ in place of $\readperm_1$ or $\readperm_2$ if the subscript is irrelevant.

In general, as defined in \Cref{sec:confluence}, we allow $\writeperm$ to split into $k$ tokens of $\readperm$, for a predetermined value of $k$, and we have one permission token for every memory region, such as the array \texttt{A} or \texttt{B}.
The intuition is that to write to a memory location, an operator needs to have exclusive write ownership of that location and no other operator should have $\writeperm$ or $\readperm$ permissions;
while to read a memory location, we allow potentially $k$ parallel reads to happen independently.
If all reads are done, and a write needs to occur, we merge all $\readperm$ tokens into one $\writeperm$ token.

These permission tokens are passed between operators via channels, and operators are not allowed to duplicate or create new tokens (i.e., permission tokens are used linearly).
We attach permissions to values flowing through the dataflow program, instead of creating separate channels for communicating permission tokens.

To determine the exact assignment and flow of permission tokens, we first attach a permission variable (representing a permission token that we do not know a priori) to each value in the channels of dataflow configurations in each symbolic branch.
\Cref{fig:example-dataflow-cut-point-2-to-cut-point-2} shows an example of this attachment of permission variables and permission constraints.
At the starting state cut point 2, we attach free permission variables $p_1, p_2, p_3, p_4$ to values in the configuration.
After the first step in which operators \exampledfphii and \exampledfphilso fire, these permission variables are consumed by the two carry operators and we generate fresh permission variables $p_5, p_6, p_7$ for the new output values.
We require that the permission variables are used in a linear fashion: operators cannot duplicate or generate permission tokens; instead, they have to be obtained from the inputs.
So for the first step in \Cref{fig:example-dataflow-cut-point-2-to-cut-point-2}, we add two constraints $p_5 + p_6 \permleq p_1 + p_3$ and $p_7 \permleq p_2 + p_4$ to say that the output permissions have be contained in the input permissions.
Furthermore, if the operator is a load or a store, such as in \Cref{fig:example-dataflow-cut-point-2-to-cut-point-2} steps 4 and 6, we require that the suitable permission ($\readperm$ for load, and $\writeperm$ for store) is present in the input permissions (e.g. $\writeperm \permleq p_{13} + p_{17 }$ at step 6 in \Cref{fig:example-dataflow-cut-point-2-to-cut-point-2}).

Finally, when the execution is finished and has either hit a final state or another cut point, we have collected a list of constraints involving a set of permission variables.
If the configuration is matched to a cut point, we add additional constraints to make sure that the assignment of tokens at the cut point is consistent.
For example in \Cref{fig:example-dataflow-cut-point-2-to-cut-point-2}, since the last configuration is matched again to the same cut point 2, we would add these additional constraints: $p_1 = p_{19}$, $p_2 = p_{18}$, $p_3 = p_8$, $p_4 = p_9$.
We then solve for the satisfiability of these constraints with respect to the permission algebra $(\{ 0, \readperm_1, \readperm_2, \writeperm \}, +, \permleq)$ that we have defined above.

If there is a satisfying assignment of the permission variables with permission tokens, then intuitively, the operators in the dataflow program are consistent in memory accesses, and there will not be any data races.
In \Cref{sec:confluence}, we show that this implies that any possible schedule will converge to the canonical schedule we use.
For the example in \Cref{fig:example-dataflow-cut-point-2-to-cut-point-2}, a satisfying assignment is $p_2 = p_7 = p_{15} = p_{16} = p_{17} = p_{18} = \writeperm$ with all other variables set to $0$.
This assignment essentially passes the write permission across the loop of operators $\exampledfphilso \to \exampledfsteerlso \to \exampledfld \to \exampledfincai \to \exampledfst \to \exampledfphilso$.

\section{Equivalence on a Canonical Schedule via Cut-Simulation}
\label{sec:sim-canonical}
We detail the first half of our translation validation technique, which
proves via \emph{cut-simulation} that the LLVM program is equivalent to the
dataflow program when restricted to a canonical schedule of dataflow operators.

Most of our formulation is independent of the source language.
Hence, instead of using LLVM directly, we assume some sequential and imperative input language called \emph{\seqlang} such that 1)~\seqlang is deterministic; 2)~\seqlang has a sequentially consistent memory model (as do dataflow programs); 3)~\seqlang programs operate on fixed-width integer values (as do dataflow programs).

Throughout this section, we fix a dataflow program and a \seqlang program, and we denote their operational semantics as transition systems 
$(\initialconfig_\df, \allconfig_\df, \to_\df)$ and $(\initialconfig_\seq, \allconfig_\seq, \to_\seq)$, respectively, where
$\initialconfig_\df$ and $\initialconfig_\seq$ are the sets of initial configurations, and $\allconfig_\df$ and $\allconfig_\seq$ are the sets of all
configurations.

We use $\mem_\df(\config_\df)$ to denote the state of the memory in a
dataflow configuration $\config_\df$, and similarly $\mem_\seq(\config_\seq)$ for
the memory in a \seqlang configuration $\config_\seq$.
We will omit the subscripts when it is clear from the context.
We assume that the memory map is modeled in the same way in both semantics, so that we can directly compare them by $\mem(\config_\df) = \mem(\config_\seq)$.
Additionally, both the dataflow and \seqlang programs have input parameters
(e.g., the pointer \texttt{\%A} in \Cref{fig:example}); thus, 
we assume a relation $\sim_\initialconfig\ \subseteq
\initialconfig_\df \times \initialconfig_\seq$ on initial configurations that
guarantees equal values on all input parameters.  

For a transition relation $\to\ \in \{ \to_\df, \to_\seq \}$, we use $\to^*$ for the reflexive-transitive closure of $\to$, and $\to^+$ for the transitive closure of $\to$.
A configuration $\config$ is \emph{terminating/final} if there exists no configuration $\config'$ such that $\config \to \config'$.
A \emph{trace} is a finite or infinite sequence of configurations $\tau = \config_1 \dots \config_i \dots$ such that $\config_i \to \config_{i + 1}$ for any $i$.
A trace $\tau$ is \emph{complete} if $\tau$ is infinite or the last configuration of $\tau$ is final.
We define $\tracesdf$ and $\tracesseq$ to be the set of all traces in the dataflow and \seqlang semantics, respectively.
For the dataflow semantics, we use $\to^o_\df$ to denote the deterministic firing of a particular operator $o$.


One important tool that we use is symbolic execution,
and in the following, we establish some notation for symbolic execution.
Let $\varlist{x}$ be a list of variables $\varlist{x} = (x_1, \dots, x_n)$.
Let $V$ be the set of all bit-vectors of a predetermined length (for modeling fixed-width integers).
We use $\term(\varlist{x})$ to denote the set of \emph{symbolic expressions}
constructed from constants in $V$, variables in $\varlist{x}$, and any
bit-vector operation we could perform in our semantics (e.g., a 32-bit zero
bit-vector $\texttt{0}_{32}$, or a term like $x_1 + \texttt{1}_{32}$ are in $\term(\varlist{x})$).
For $s \in \{ \df, \seq \}$, we use $\allconfig_s(\varlist{x})$ to represent the set of \emph{symbolic configurations with free variables $\varlist{x}$} in the semantics $s$.
A symbolic configuration $\config(\varlist{x}) \in \allconfig_s(\varlist{x})$ is
similar to a configuration but over symbolic expressions $\term(\varlist{x})$ instead of only $V$ (for example, a symbolic dataflow configuration can have symbolic expressions such as $x_1 + \texttt{1}_{32}$ in the channels).
If we have a substitution $\sigma: \{ x_1, \dots, x_n \} \to \term(\varlist{y})$ from $\varlist{x}$ to symbolic expressions over some other list of variables $\varlist{y}$, we denote $\config(\sigma) \in \allconfig_s(\varlist{y})$ to be the result of the substitution.
We use $\config(\varlist{x}) \to_s \config'(\varlist{x}) \mid \varphi(\varlist{x})$ to denote symbolic execution from $\config(\varlist{x})$ to $\config'(\varlist{x})$ with path condition $\varphi(\varlist{x})$.

A dataflow program does not have a notion of a return value, and the only observable state after finishing execution is the modified global memory.
Thus we first define when two traces are equivalent based on their memory state.

\begin{definition}[Memory-synchronizing traces]
Let $\tau_\df = \config^1_\df \dots \config^i_\df \dots$ and $\tau_\seq = \config^1_\seq \dots \config^j_\seq \dots$ be two complete traces starting from $\config^1_\df$ and $\config^1_\seq$, respectively.
The traces $\tau_\df$ and $\tau_\seq$ are \emph{memory-synchronizing} iff $\mem(\config^1_\df) = \mem(\config^1_\seq)$, and
\begin{itemize}
    \item If $\tau_\df$ is finite, then $\tau_\seq$ is finite and $\mem(\config'_\df) = \mem(\config'_\seq)$, where $\config'_\df$ and $\config'_\seq$ are the final states of $\tau_\df$ and $\tau_\seq$, respectively.
    \item If $\tau_\df$ is infinite, then $\tau_\seq$ is infinite, and for any $n \in \mathbb{N}$, there are $i, j \ge n$ such that $\mem(\config^i_\df) = \mem(\config^j_\seq)$.
\end{itemize}
\end{definition}
In other words, two complete dataflow and \seqlang traces are memory-synchronizing if they are both finite and have the same memory state in the initial and final configurations (i.e., both programs terminate and have the same final memory state); or they are both infinite and they synchronize in the same memory state infinitely often.

Now, we may define \emph{program equivalence}, which describes when two matching
initial configurations for dataflow and \seqlang exhibit the same observable
behaviors:
\begin{definition}[Program equivalence]
\label{def:full-equivalence}
Two programs $(\initialconfig_\df, \allconfig_\df, \to_\df)$ and $(\initialconfig_\seq, \allconfig_\seq, \to_\seq)$ are \emph{equivalent} if for any pair of initial states 
$\config_\df \sim_\initialconfig \config_\seq$ 
with $\mem(\config_\df) = \mem(\config_\seq)$, if we take any complete dataflow trace $\tau_\df = \config_\df \dots$ and the unique complete \seqlang trace $\tau_\seq = \config_\seq \dots$ (unique since \seqlang is deterministic), then $\tau_\df$ and $\tau_\seq$ are memory-synchronizing.
\end{definition}
Intuitively, the dataflow program and \seqlang program are equivalent if any pair of complete traces from corresponding initial states are memory-synchronizing; i.e., they either both terminate with the same memory state, or both run forever and synchronize infinitely often. 

Our final goal is to achieve \Cref{def:full-equivalence}.
We will do so by first verifying a weaker property of equivalence
along a particular, deterministic \emph{canonical} schedule of the dataflow
program. Then, in \Cref{sec:confluence}, we show that this canonical schedule
suffices to prove equivalence for all possible dataflow schedules. 

First, we introduce canonical schedules, which determinize dataflow programs using
the information contained in Seq traces:

\begin{definition}[Canonical schedule]
\label{def:canonical-schedule}
A \emph{canonical schedule} is a function $\canonsched: \allconfig_\df \times \tracesseq \to \tracesdf$ that selects a dataflow trace using a \seqlang trace, such that
\begin{itemize}
    \item For any $\config_\df \in \allconfig_\df$ and any non-empty \seqlang trace $\tau_\seq$, $\canonsched(\config_\df, \tau_\seq)$ is a dataflow trace starting in $\config_\df$.
    \item If $\tau_\seq$ is finite, then $\canonsched(\config_\df, \tau_\seq)$ is finite.
    \item For any $\config_\df \in \allconfig_\df$, if $\tau'_\seq$ is a prefix of $\tau_\seq$, then $\canonsched(\config_\df, \tau'_\seq)$ is a prefix of $\canonsched(\config_\df, \tau_\seq)$.
\end{itemize}
\end{definition}

We then specialize program equivalence (\Cref{def:full-equivalence}) by
restricting our attention to those canonical dataflow schedules:
\begin{definition}[Program equivalence on a canonical schedule]
\label{def:canonical-equivalence}
Let $\canonsched$ be a canonical schedule.
Two programs $(\initialconfig_\df, \allconfig_\df, \to_\df)$ and
    $(\initialconfig_\seq, \allconfig_\seq, \to_\seq)$ are \emph{equivalent on
    the canonical schedule $\canonsched$} if for any pair of initial states
    $\config_\df \sim_\initialconfig \config_\seq$ with $\mem(\config_\df) = \mem(\config_\seq)$, if we take the unique complete \seqlang trace $\tau_\seq = \config_\seq \dots$,
then $\canonsched(\config_\df, \tau_\seq)$ is a complete dataflow trace, and $\canonsched(\config_\df, \tau_\seq)$ and $\tau_\seq$ are memory-synchronizing.


\end{definition}
That is, instead of enforcing that all complete dataflow traces are memory-synchronizing to the complete \seqlang trace, we only require that there exists one such trace selected by some canonical schedule $\canonsched$.
Later in \Cref{sec:confluence}, we separately verify the \emph{confluence} property (\Cref{def:confluence}), so that \Cref{def:confluence} and \Cref{def:canonical-equivalence} together imply \Cref{def:full-equivalence} (\Cref{prop:full-equivalence}).

\paragraph{Canonical Schedules for RipTide}
\phantomsection
\label{example:canonical-schedule}
In our case study of verifying compilation from LLVM programs to RipTide dataflow programs, the specific canonical schedule $\canonsched$ is a heuristic designed for LLVM and the RipTide compiler.
We instrument the RipTide compiler to output a partial map $\hint: \{ I_1, \dots, I_m \} \rightharpoonup \{ o_1, \dots, o_n \}$, where $I_1, \dots, I_m$ are LLVM instructions in the LLVM program, and $o_1, \dots, o_n$ are operators in the dataflow program.
We require that the image of $\hint$ covers all dataflow operators except for steer operators (which are control-flow operators and have no direct correspondence to an LLVM instruction).
The canonical schedule $\canonsched(\config_\df, \tau_\llvm)$ is then inductively constructed following the LLVM instructions executed in $\tau_\llvm$:
\begin{equation*}
\canonsched(\config_\df, \tau_\llvm) := \begin{cases}
    \epsilon & \text{If } \tau_\llvm = \epsilon \\
    \tau_\df \cdot \canonsched(\config'_\df, \tau'_\llvm) & \text{If } \tau_\llvm = \config_\llvm \cdot \tau'_\llvm
\end{cases}
\end{equation*}
where $\epsilon$ is the empty trace, $\cdot$ denotes trace concatenation, $I$ is
the LLVM instruction executed at $\config_\llvm$, and $\tau_\df = \config_\df
\dots \config'_\df$ is a dataflow trace firing $\hint(I)$ and then any fireable
steer operators (which may be fired in any order, as they all commute).
If $\hint(I)$ is not defined or not fireable in $\config_\df$, then $\tau_\df = \epsilon$ and $\config'_\df = \config_\df$.

In the following sections, we outline the steps to verify that the dataflow and \seqlang programs satisfy \Cref{def:canonical-equivalence}.
In \Cref{sec:equiv-via-sim}, we first describe our general method of using a coinductive invariant called cut-simulation that implies the property in \Cref{def:canonical-equivalence}.
In \Cref{sec:sim-cut-points}, we show how cut-simulation can be symbolically represented by cut points.
In \Cref{sec:sim-check}, we present the algorithm to check that a set of cut points does form a cut-simulation.
Finally in \Cref{sec:sim-infer}, we describe our heuristics to infer cut points for the compilation from LLVM to dataflow programs.

\subsection{Proving Equivalence on the Canonical Schedule via Cut-Simulation}
\label{sec:equiv-via-sim}

Traces from $\initialconfig_\df$ and $\initialconfig_\seq$ are unbounded in general, so to prove the equivalence on a canonical schedule (\Cref{def:canonical-equivalence}),
we need to verify a coinductive invariant about the pairs of traces starting from $(\initialconfig_\df, \initialconfig_\seq)$, and such an invariant is a \emph{cut-simulation} relation~\cite{keq}.

\begin{definition}[Cut-simulation]
Let $\mathcal{R} \subseteq \allconfig_\df \times \allconfig_\seq$ be a binary relation.
We say that $\mathcal{R}$ is a \emph{cut-simulation} between $(\initialconfig_\seq, \allconfig_\seq, \to_\seq)$ and $(\initialconfig_\df, \allconfig_\df, \to_\df)$ iff for any $(\config_\df, \config_\seq) \in \mathcal{R}$, if $\config_\seq \to_\seq \config'_\seq$ for some $\config'_\seq$, then there are $(\config'_\df, \config''_\seq) \in \mathcal{R}$ such that $\config'_\seq \to^*_\seq \config''_\seq$ and $\config_\df \to^+_\df \config'_\df$.
\end{definition}

Intuitively, a cut-simulation states that if the two transition systems are
synchronized in $\mathcal{R}$, then if we can take a step on the \seqlang side, we can step the two transition systems further to re-synchronize them in $\mathcal{R}$.

Note that in our setting, we are verifying that the dataflow program simulates the \seqlang program (i.e., the target program simulates the source program), instead of the other way around which is common in compiler verification work.
This is because our target dataflow program has more nondeterministic behavior than the source \seqlang program, whereas in many other cases, the source program has more behavior (such as undefined behavior in C) than the concrete target program.


To satisfy our desired equivalence properties, we look for cut-simulations that are
\emph{memory-synchronizing}:
\begin{definition}
\label{def:memory-sync}
A binary relation $\mathcal{R} \subseteq \allconfig_\df \times \allconfig_\seq$ is \emph{memory-synchronizing} iff
\begin{itemize}
    \item For any $(\config_\df, \config_\seq) \in \mathcal{R}$, $\mem(\config_\df) = \mem(\config_\seq)$.
    \item For any initial configurations $\config_\df \sim_\initialconfig \config_\seq$ with $\mem(\config_\df) = \mem(\config_\seq)$, $(\config_\df, \config_\seq) \in \mathcal{R}$.
    \item For any $(\config_\df, \config_\seq)$ with $\config_\df, \config_\seq$ final and $\mem(\config_\df) = \mem(\config_\seq)$, $(\config_\df, \config_\seq) \in \mathcal{R}$.
\end{itemize}
\end{definition}
The second and third conditions enforce that $\mathcal{R}$ should include all pairs of corresponding initial states, as well as final states; thus $\mathcal{R}$ covers all complete traces.

If there exists a memory-synchronizing cut-simulation $\mathcal{R}$,
then we can conclude via a coinductive argument that the two programs are equivalent on a canonical schedule (\Cref{def:canonical-equivalence}).

\subsection{Describing a Cut-Simulation via Cut Points}
\label{sec:sim-cut-points}

To finitely describe a proposed cut-simulation, we use a list of pairs of dataflow and \seqlang \emph{cut points}.
A cut point is a pair $\cutpoint = (\config_\df(\varlist{x}), \config_\seq(\varlist{y})) \mid \varphi(\varlist{x}, \varlist{y})$ of symbolic dataflow and \seqlang configurations constrained by the correspondence condition $\varphi(\varlist{x}, \varlist{y})$.
In our case, the condition $\varphi(\varlist{x}, \varlist{y})$ is a conjunction of equalities between variables in $\varlist{x}$ and $\varlist{y}$.

Semantically, each cut point describes the binary relation $\mathcal{R}(\cutpoint) := \{ (\config_\df(\varlist{x}), \config_\seq(\varlist{y})) \in \allconfig_\df \times \allconfig_\seq \mid \varlist{x}, \varlist{y} \in V, \varphi(\varlist{x}, \varlist{y}) \}$.
For a list of $n$ cut points
\[ \{\cutpoint^i = (\config^i_\df(\varlist{x^i}),
\config^i_\seq(\varlist{y^i})) \mid \varphi^i(\varlist{x^i}, \varlist{y^i})
\}_{i = 1}^n, \] 
the relation they describe is the union of the subrelation each describes: $\mathcal{R}(\cutpoint^1, \dots, \cutpoint^n) := \bigcup^n_{i = 1} \mathcal{R}(\cutpoint^i)$.

To ensure that the cut-simulation we propose satisfies the three conditions in \Cref{def:memory-sync}, we only construct cut points such that:
\begin{itemize}
    \item For any cut point $\cutpoint^i$, the correspondence constraint $\varphi^i(\varlist{x}^i, \varlist{y}^i)$ implies equal memory $\mem(\config^i_\df(\varlist{x}^i)) = \mem(\config^i_\seq(\varlist{y}^i))$.
    \item The first cut point exactly represents the pairs of equivalent,
        initial configurations with equal memory; i.e.,
        $\mathcal{R}(\cutpoint^1) = \{ (\config_\df, \config_\seq) \in\ 
        \sim_\initialconfig\ \mid \mem(\config_\df) = \mem(\config_\seq) \}$.
    \item The last cut point $\cutpoint^n$ exactly represents the pairs of final configurations with equal memory; i.e., $\mathcal{R}(\cutpoint^n) = \{ (\config_\df, \config_\seq) \in \allconfig_\df \times \allconfig_\seq \mid \config_\df, \config_\seq \text{ final}, \mem(\config_\df) = \mem(\config_\seq) \}$.
\end{itemize}
Then $\mathcal{R}(\cutpoint^1, \dots, \cutpoint^n)$ satisfies \Cref{def:memory-sync}.

\subsection{Checking Cut-Simulation}
\label{sec:sim-check}

Now that we have a way to symbolically describe a relation $\mathcal{R}(\cutpoint^1, \dots, \cutpoint^n)$ with cut points satisfying \Cref{def:memory-sync},
we need to check that it is indeed a cut-simulation.

To achieve this, we perform symbolic execution from each cut point and try to reach another cut point.
Suppose we are checking cut point $\cutpoint^i = (\config^i_\df(\varlist{x^i}), \config^i_\seq(\varlist{y^i})) \mid \varphi^i(\varlist{x^i}, \varlist{y^i})$.
We perform the following three steps.

\paragraph{Step 1. Symbolic execution of \seqlang cut points}
We perform symbolic execution from $\config^i_\seq(\varlist{y^i})$ and get a list of $k_i$ symbolic branches with path conditions denoted by $\psi$:
\begin{gather*}
    \config^i_\seq(\varlist{y^i}) \to^+_\seq \config^{i, 1}_\seq(\varlist{y^i}) \mid \psi_{\seq, 1}(\varlist{y^i}) \\ 
    \dots \\
    \config^i_\seq(\varlist{y^i}) \to^+_\seq \config^{i, k_i}_\seq(\varlist{y^i}) \mid \psi_{\seq, k_i}(\varlist{y^i})
\end{gather*}
where each right hand side $\config^{i, l}_\seq(\varlist{y^i})$ for $l \in \{ 1,
\dots, k_i \}$ is an instance $\config^j_\seq(\sigma)$ of some \seqlang cut
point $\config^j_\seq(\varlist{y^j})$ with substitution $\sigma: \varlist{y^j}
\to \term(\varlist{y^i})$.
That is, at the end of each branch, the configuration is matched to some cut point.

Note that since the initial symbolic configuration $\config^i_\seq(\varlist{y^i})$ does not have any constraints, the final path conditions should cover all cases of $\varlist{y^i}$; in other words, $\bigvee^{k_i}_{l = 1} \psi_{\seq, l}(\varlist{y^i})$ is valid.

\paragraph{Step 2. Symbolic execution of dataflow cut points on the canonical schedule}
For the dataflow cut point $\config^i_\df(\varlist{x^i})$, since the dataflow semantics is nondeterministic, the choice of operators to fire is unclear.
Therefore, we follow a canonical schedule $\canonsched$ (\Cref{def:canonical-schedule}) to replicate the execution of \seqlang branches on the dataflow side.

For each \seqlang branch $l \in \{ 1, \dots, k_i \}$, we have a \seqlang trace $\tau_\seq = \config^i_\seq(\varlist{y}^i) \dots \config^{i, l}_\seq(\varlist{y}^i)$.
By mapping this trace through the canonical schedule $\tau_\df := \canonsched(\config^i_\df(\varlist{x^i}), \tau_\seq)$, we get a list of dataflow operators $(o^l_1, \dots, o^l_{n_l})$ fired in $\tau_\df$.
Then we symbolically execute the dataflow cut point $\config^i_\df(\varlist{x^i})$ on the schedule $(o^l_1, \dots, o^l_{n_l})$ with the path condition
$\psi_{\df, l}(\varlist{x^i}) := \exists \varlist{y^i} \varphi^i(\varlist{x^i}, \varlist{y^i}) \wedge \psi_{\seq, l}(\varlist{y^i})$:
\[
    \config^i_\df(\varlist{x^i}) \to^{o^l_1}_\df \dots \to^{o^l_{n_l}}_\df \config^{i, l}_\df(\varlist{x^i}) \mid \psi_{\df, l}(\varlist{x^i})
\]
The new path condition $\psi_{\df, l}(\varlist{x^i})$ amounts to replacing variables in the \seqlang path condition $\psi_{\seq, l}(\varlist{y^i})$ with corresponding variables in the correspondence constraint $\varphi^i(\varlist{x^i}, \varlist{y^i})$.
With the additional path condition imposed, the symbolic execution is not likely
to branch. If it does, the cut-simulation check will fail; in which case, a more precise
heuristic is needed for the canonical schedule.

By doing this for each \seqlang branch, we get a list of dataflow branches:
\begin{gather}
    \config^i_\df(\varlist{x^i}) \to^+_\df \config^{i, 1}_\df(\varlist{x^i}) \mid \psi_{\df, 1}(\varlist{x^i}) \label{stmt:dataflow-symbolic-execution} \\
    \dots \nonumber \\
    \config^i_\df(\varlist{x^i}) \to^+_\df \config^{i, k_i}_\df(\varlist{x^i}) \mid \psi_{\df, k_i}(\varlist{x^i}) \nonumber
\end{gather}
In our case, these branches should cover all possible values of $\varlist{x}^i$
(i.e., $\bigvee^{k_i}_{l = 1} \psi_{\df, l}(\varlist{x^i})$ is valid), since in
the path condition $\psi_{\df, l}(\varlist{x^i})$, the correspondence
constraint $\varphi^i(\varlist{x^i}, \varlist{y^i})$ only contains a conjunction of equalities between variables in $\varlist{x^i}$ and $\varlist{y^i}$,
and the union of all \seqlang path conditions $\bigvee^{k_i}_{l = 1} \psi_{\seq, l}(\varlist{y^i})$ is valid.
In general, however, if the correspondence constraint $\varphi^i(\varlist{x^i}, \varlist{y^i})$ is more complex (e.g. containing formulas such as $x_1 = 2 y_2$), we may need to additionally verify that $\bigvee^{k_i}_{l = 1} \psi_{\df, l}(\varlist{x^i})$ is valid to ensure that all possible concrete configurations represented by $\config^i_\df(\varlist{x^i})$ is covered in the symbolic execution.

\paragraph{Step 3. Cut point subsumption}
Finally, we check that each pair of dataflow and \seqlang branches $(\config^{i, l}_\df(\varlist{x^i}), \config^{i, l}_\seq(\varlist{y^i}))$
is contained in some target cut point $\cutpoint^j = (\config^j_\df(\varlist{x^j}), \config^j_\seq(\varlist{y^j})) \mid \varphi^j(\varlist{x^j}, \varlist{y^j})$:
\begin{equation}
    \label{stmt:cut-point-subsumption}
    \{ (\config^{i, l}_\df(\varlist{x^i}), \config^{i, l}_\seq(\varlist{y^i})) \mid \varphi^i(\varlist{x^i}, \varlist{y^i}) \wedge \psi_{\df, l}(\varlist{x^i}) \wedge \psi_{\seq, l}(\varlist{y^i}) \} \subseteq \mathcal{R}(\cutpoint^j)
\end{equation}
or equivalently, there are substitutions $\sigma^{i, l}_\df: \varlist{x^j} \to
\term(\varlist{x^i}), \sigma^{i, l}_\seq: \varlist{y^j} \to \term(\varlist{y^i})$ such that the following conditions hold:
\begin{itemize}
    \item The dataflow branch is an instance of the target dataflow cut point: $\config^{i, l}_\df(\varlist{x^i}) = \config^j_\df(\sigma^{i, l}_\df)$.
    \item The \seqlang branch is an instance of the target \seqlang cut point: $\config^{i, l}_\seq(\varlist{y^i}) = \config^j_\seq(\sigma^{i, l}_\seq)$.
    \item The correspondence condition at the target cut point is satisfied given the source cut point correspondence and the path conditions:
    \[
        \varphi^i(\varlist{x^i}, \varlist{y^i}) \wedge \psi_{\df, l}(\varlist{x^i}) \wedge \psi_{\seq, l}(\varlist{y^i}) \rightarrow \varphi^j(\sigma^{i, l}_\df, \sigma^{i, l}_\seq)
    \]
\end{itemize}

If the steps in this section are passed for all cut points $i$, we can conclude that $\mathcal{R}(\cutpoint^1, \dots, \cutpoint^n)$ is indeed a cut-simulation.



\subsection{Heuristics to Infer Cut Points}
\label{sec:sim-infer}

In this section, we describe our heuristics to infer cut points specific to LLVM and dataflow programs.
In general, program equivalence for Turing-complete models like LLVM and dataflow programs is undecidable,
so our best hope is to have good heuristics to infer the cut points for the specific input/output programs of the compiler.

The placement of cut points is important for the symbolic execution to terminate.
It is in theory correct (for terminating dataflow and LLVM programs) to have exactly two cut points: one at the initial configurations, and one for all final configurations.
But the symbolic execution will try to simulate all possible traces together, which are unbounded, and the check may not terminate.
To ensure termination, on the LLVM side, we place a cut point at the entry of
the function (i.e., the initial configuration), the back edge of each loop
header block, and the exit program point. 




For the dataflow program, we need to place the first cut point at the initial configuration of the dataflow program, and the last cut point at the final configuration.
For loops, ideally, we want to place the cut points at the corresponding ``program points'' at loop headers.
However, due to the lack of control-flow in the dataflow program, the form of the corresponding dataflow configurations at these ``program points'' is hard to infer statically.
Instead of generating the dataflow cut points ahead of time, we dynamically infer them during the symbolic execution of dataflow branches.
For example, after mirroring the LLVM branch execution in a dataflow branch such
as in \Cref{stmt:dataflow-symbolic-execution}, if the dataflow cut point
corresponding to the LLVM cut point has not been inferred yet, we generalize the
symbolic configuration on the RHS of the dataflow branch
\Cref{stmt:dataflow-symbolic-execution} to a dataflow cut point by replacing symbolic expressions with fresh variables.
The correspondence condition is then inferred using compiler hints to match fresh variables in the new dataflow cut point to an LLVM variable defined by some LLVM instruction.

Note that we do not need these heuristics to be sound.
As long as the cut points for the initial and final dataflow/LLVM configurations are generated correctly and the cut-simulation check in \Cref{sec:sim-check} succeeds, the soundness of the cut-simulation follows.
While we do not prove that these heuristics are complete, they work well in practice and cover all test cases in our evaluation (\Cref{sec:evaluation}).

If the heuristics generate incorrect cut points, the simulation check would still terminate but fail.
To see this, on the LLVM side, our placement of the cut points guarantees termination because we place one at each loop header.
On the dataflow side, since the execution is mirroring the LLVM execution and no additional search is done, the algorithm would terminate but may fail to match the resulting dataflow configuration against the cut point.



\section{Confluence Checking using Linear Permission Tokens}
\label{sec:confluence}

After checking the equivalence between the LLVM and dataflow programs on
a canonical dataflow schedule, we still need to prove that this canonical schedule is 
general.
Indeed, since the dataflow program is asynchronous and nondeterministic, other choices of
schedule may lead to data races and inconsistent final states.

Therefore, in addition to checking the equivalence on a canonical schedule, we want to prove a \emph{confluence} property on the dataflow program:
\begin{definition}
\label{def:confluence}
We say $(\initialconfig_\df, \allconfig_\df, \to_\df)$ is \emph{(ground) confluent} iff for any initial configuration $\config \in \initialconfig_\df$, if $\config \to^*_\df \config_1$ and $\config \to^*_\df \config_2$, then there is $\config' \in \allconfig_\df$ with $\config_1 \to^*_\df \config'$ and $\config_2 \to^*_\df \config'$.
\end{definition}
If the pair of dataflow and LLVM programs satisfies both \Cref{def:canonical-equivalence} and \Cref{def:confluence}, then we can conclude the desired equivalence property of \Cref{def:full-equivalence}.

\begin{proposition}
\label{prop:full-equivalence}
If $(\initialconfig_\df, \allconfig_\df, \to_\df)$ and $(\initialconfig_\llvm, \allconfig_\llvm, \to_\llvm)$ are equivalent on a canonical schedule (\Cref{def:canonical-equivalence}) and $(\initialconfig_\df, \allconfig_\df, \to_\df)$ is confluent, then $(\initialconfig_\df, \allconfig_\df, \to_\df)$ and $(\initialconfig_\llvm, \allconfig_\llvm, \to_\llvm)$ are equivalent (\Cref{def:full-equivalence}).
\end{proposition}
\begin{proof}[Proof Sketch]
Assume $(\initialconfig_\df, \allconfig_\df, \to_\df)$ and $(\initialconfig_\llvm, \allconfig_\llvm, \to_\llvm)$ are equivalent on some canonical schedule $\canonsched$ and $(\initialconfig_\df, \allconfig_\df, \to_\df)$ is confluent.
Let $\config_\df \sim_\initialconfig \config_\seq$ be any pair of corresponding initial states with $\mem(\config_\df) = \mem(\config_\seq)$, and let $\tau_\seq = \config_\seq \dots$ be the unique complete trace.
Let $\tau'_\df = \config_\df \dots$ be any complete trace from $\config_\df$.
We need to show that $\tau'_\df$ and $\tau_\seq$ are memory-synchronizing.

By equivalence on $\canonsched$, we know that $\canonsched(\config_\df, \tau_\seq)$ is a complete trace and $\canonsched(\config_\df, \tau_\seq)$ and $\tau_\seq$ are memory-synchronizing.
We have two cases:
\begin{itemize}
    \item If $\canonsched(\config_\df, \tau_\seq)$ is finite, then by confluence, $\tau'_\df$ is finite and has the same final state as $\canonsched(\config_\df, \tau_\seq)$ (since $\tau'_\df$ is complete).
    \item If $\canonsched(\config_\df, \tau_\seq)$ is infinite, then by confluence, $\canonsched(\config_\df, \tau_\seq)$ and $\tau'_\df$ synchronize infinitely often.
\end{itemize}
Therefore $\tau'_\df$ and $\tau_\seq$ are memory-synchronizing.
\end{proof}

In this section, we present a technique using what we call \emph{linear permission tokens} to ensure that the dataflow program is confluent.
Linear permission tokens are, in essence, a dynamic variant of fractional permissions~\cite{fractional-perm}.
In our case, the exact placement of permission tokens is resolved dynamically via symbolic execution.

We associate each value flowing through a dataflow program with a permission token in the form of $\writeperm\ l$, $\readperm\ l$, or any formal sum of them, for any memory location $l$.
These tokens are used by dataflow operators in a linear fashion: they cannot be replicated or generated, and they are disjoint in the initial state.
To support independent parallel reads and enforce that writes must be exclusive, a \writeperm\ token can be split into $k$ many \readperm\ tokens (for a predetermined parameter $k$), and inversely, $k$ of such \readperm\ tokens can be merged to get back a \writeperm\ token.

We show in this section that if there is a consistent assignment of these linear permission tokens to values, then we can conclude that the execution of dataflow program will be confluent.
This is done in two steps:
\begin{itemize}
    \item In \Cref{sec:confluence-bounded}, we prove that the technique works in a bounded setting for a concrete trace $\tau$: if $\tau$ has a consistent assignment of permission tokens, $\tau'$ is another trace from the same initial configuration, and $\tau'$ is bounded by $\tau$ (meaning that for each operator $o$, the number of times $o$ fired in $\tau'$ is less than or equal to the number of times $o$ fired in $\tau$), then there is a valid trace from the final configuration of $\tau'$ to the final configuration of $\tau$; i.e., $\tau'$ will converge to the final configuration of $\tau$.
    \item In \Cref{sec:confluence-full}, we explain how this bounded confluence check can be extended to a full confluence check for the entire dataflow program by using the cut points that we have placed for the dataflow program in \Cref{sec:sim-canonical}.
\end{itemize}
Finally in \Cref{sec:bounded-channel}, we show that even though we assume the channels in a dataflow program are unbounded, the cut-simulation and confluence checks allow us to extend the verification results to a bounded-channel model of dataflow programs.

\subsection{Confluence on Bounded Traces}
\label{sec:confluence-bounded}


In this section, we show that if we have a finite trace $\config_1 \to^{o_1}_\df \dots \to^{o_n}_\df \config_{n + 1}$ (where $o_1, \dots, o_n$ are the operators fired in each step) with a consistent assignment of permission tokens to each value in each configuration, then for any other trace $\config_1 \to^{o'_1}_\df \dots \to^{o'_m}_\df \config_{m + 1}$ (not necessarily with a consistent assignment of permission tokens) such that the multiset $\{ o'_1, \dots, o'_m \} \subseteq \{ o_1, \dots, o_n \}$, we have $\config_{m + 1} \to^*_\df \config_{n + 1}$.

Let us first describe the permission tokens and their assignment to values.
The permission tokens are drawn from the \emph{permission algebra} defined below.
\begin{definition}[Permission Algebra]
\label{def:permission-algebra}
Let $L$ be a set of memory locations and $k$ be a positive integer.
The \emph{permission algebra} is a partial algebraic structure $(P(L, k), 0, +, \permleq)$ where
\begin{itemize}
    \item $P(L, k)$ is the set of permissions such that $(P(L, k), 0, +)$ forms a partially commutative monoid.
    It consists of formal sums $\readperm\ l_1 + \dots + \readperm\ l_n$, with $l_1, \dots, l_n \in L$, modulo the commutativity and associativity of $+$.
    We require that for a location $l$, at most $k$ copies of $\readperm\ l$ are present in the formal sum.
    We denote $N(l, p)$ to be the number of occurrences of $\readperm\ l$ in $p$.
    We call $p + q$ the \emph{(disjoint) sum} of $p, q \in P(L, k)$.
    In particular, for any $l \in L$, $p + q$ is defined iff for all $l$, $N(l, p) + N(l, q) \leq k$.

    We also denote
    \[
        \writeperm\ l := \underbrace{\readperm\ l + \dots + \readperm\ l}_{k \text{ copies}}
    \]

    \item $P(L, k)$ is partially ordered by $\permleq$ with $p \permleq q$ (reads $q$ \emph{contains} $p$) iff for any $l$, the number of occurrences of $\readperm\ l$ in $p$ is less than or equal to that in $q$.
\end{itemize}
\end{definition}


Now we define when two or more permissions are considered \emph{disjoint}, which is crucial for enforcing linearity.
\begin{definition}
We say that the permissions $p_1, \dots, p_n \in P(L, k)$ are \emph{disjoint} (write $\permdisjoint(p_1, \dots, p_n)$) iff
for any $l \in L$ and $i \ne j$, $N(l, p_i) + N(l, p_j) \leq k$.
\end{definition}

For example, in $P(\{ \texttt{A}, \texttt{B} \}, 2)$, we have that 
\[ \readperm\ \texttt{B} + \readperm\ \texttt{A} + \readperm\ \texttt{A} =
\readperm\ \texttt{B} + \writeperm\ \texttt{A},\]
while $\readperm\ \texttt{A} + \writeperm\ \texttt{A}$ is not defined (i.e., $\readperm\ \texttt{A}$ and $\writeperm\ \texttt{A}$ are not disjoint).

Intuitively, the reason to allow a $\writeperm$ to split into $k$ copies of $\readperm$ is to allow the following two valid scenarios:
\begin{itemize}
    \item A store operator can write to a memory location if it has exclusive write permission to the location (i.e., no other operator has a $\writeperm$ or $\readperm$ permission to that memory location);
    \item A load operator can read from a memory location if it has a $\readperm$ permission and all other currently existing permissions in the system are also $\readperm$.
\end{itemize}
Furthermore, once (at most $k$) parallel reads are performed, a store operator can perform a memory write by reclaiming all existing $\readperm$ permissions and merging them into a $\writeperm$; hence we define $\writeperm$ as the sum of $k$ copies of $\readperm$.

Now let us define the assignment of permission tokens to values in a configuration.
Let $E$ denote the set of channel names.
For a configuration $\config \in \allconfig_\df$ and a channel $e \in E$, let $\channel(\config, e)$ denote the state of the channel $e$ in $\config$, which is a string over $V$ representing the values in the channel $e$.

\begin{definition}[Permission Augmentation]
\label{def:permission-augmentation}
Let $\config \in \allconfig_\df$ be a dataflow configuration.
A \emph{permission augmentation} of $\config$ is a partial map $t: E \times \mathbb{N} \rightharpoonup P(L, k)$ such that for any $e \in E$, $t(e, n)$ is defined iff $n \in \{ 1, \dots, |\channel(c, e)| \}$.
We call $(\config, t)$ a \emph{permission-augmented configuration}.
\end{definition}

Intuitively, $t(e, n)$ is the permission we attach to the $n$-th value in the channel $e$.
Permission augmentations defined above can assign arbitrary permissions to values.
However, to ensure race-freedom and confluence, we need to restrict them so that, for example, two operators cannot share a $\writeperm$ token.
In the following definitions, we define when a configuration, a transition, and a trace are \emph{consistently augmented with permissions}, meaning that permissions are used linearly.


\begin{definition}[Consistent Augmentation]
\label{def:consistent-augmentation}
We call $(\config, t)$ a \emph{consistent (permission-augmented) configuration} if the permissions in the image of $t$ are disjoint.
\end{definition}

\begin{definition}[Consistent Transition]
\label{def:consistent-transition}
Let $\config \to^o_\df \config'$ be a transition.
Let $t, t'$ be two consistent permission augmentations for $\config, \config'$, respectively.
Let $p_1, \dots, p_n$ be the permissions in $t$ attached to input values of $o$ and $q_1, \dots, q_m$ be the permissions in $t'$ attached to output values of $o$.
We say $(\config, t) \to^o_\df (\config', t')$ is a \emph{consistent (permission-augmented) transition} iff:
\begin{itemize}
    \item $(\config, t)$ and $(\config', t')$ are consistent.
    \item $t(e, n) = t'(e, n)$ for all channel $e$ and position $n$ except for those changed by $o$.
    \item (Linearity) $q_1 + \dots + q_m \permleq p_1 + \dots + p_n$.
    \item (Load Permission) If $o$ is a load operator on a memory location $l \in L$, then $\readperm\ l \permleq p_1 + \dots + p_n$.
    \item (Store Permission) If $o$ is a store operator on a memory location $l \in L$, then $\writeperm\ l \permleq p_1 + \dots + p_n$.
\end{itemize}
\end{definition}

\begin{definition}[Consistent Trace]
Let $\config_1 \to_\df \dots \to_\df \config_n$ be a trace.
Let $t_1, \dots, t_n$ be permission augmentations to $\config_1, \dots, \config_n$, respectively.
We say that $(\config_1, t_1) \to_\df \dots \to_\df (\config_n, t_n)$ is a \emph{consistent (permission-augmented) trace} iff for all $i \in \{ 1, \dots, n \}$, $(\config_i, t_i) \to_\df (\config_{i + 1}, t_{i + 1})$ is consistent.
\end{definition}

We now show some facts about a consistently permission-augmented transition or trace, leading to the final \Cref{prop:consistent-commute-multiset}.
First, we show that two consistent steps commute.

\begin{lemma}
\label{prop:consistent-commute}
Let $(\config_1, t_1) \to^{o_1}_\df (\config_2, t_2) \to^{o_2}_\df (\config_3, t_3)$ be a consistent trace.
Suppose $\config_1 \to^{o_2}_\df \config'_2$ for some configuration $\config'_2$.
Then $\config'_2 \to^{o_1}_\df \config_3$ and there is a permission augmentation $t'_2$ such that $(\config_1, t_1) \to^{o_2}_\df (\config'_2, t'_2) \to^{o_1}_\df (\config_3, t_3)$ is consistent.
\end{lemma}
\begin{proof}[Proof Sketch]
If one of $o_1$ and $o_2$ is neither load nor store, or if both of them are loads, they can trivially commute with the same result.
Otherwise, one of $o_1$ and $o_2$ is a store operator (without loss of generality, assume $o_1$ is a store).
For $i \in \{ 1, 2 \}$, let $l_i$ be the memory location accessed by $o_i$ and $p_i$ be the sum of $o_i$'s input permissions.
Since they can both fire at $\config_1$, they should have suitable permissions: $\writeperm\ l_1 \permleq p_1$ and $\readperm\ l_2 \permleq p_2$.
Since $p_1$ and $p_2$ are disjoint, we have $l_1 \ne l_2$
Therefore, they are accessing different memory locations, and should give the same result no matter in which order we execute them.
\end{proof}

Inductively applying this lemma, we can swap the execution of an operator all the way to the beginning if it can fire and has not fired yet.

\begin{lemma}
\label{prop:consistent-commute-n}
Let $(\config_1, t_1) \to^{o_1}_\df \dots \to^{o_n}_\df (\config_{n + 1}, t_{n + 1})$ be a consistent trace.
Suppose $o_n$ is not in the set $\{ o_1, \dots, o_{n - 1} \}$ and $\config_1 \to^{o_n}_\df \config'_2$ for some $\config'_2$,
then $\config_1 \to^{o_n}_\df \config'_2 \to^{o_1}_\df \dots \to^{o_{n - 1}}_\df \config_k$ for some $\config'_2$ and there is a consistent augmentation for this trace.
\end{lemma}

Applying the lemma above inductively, we have the following.

\begin{theorem}[Bounded Confluence for Consistent Traces]
\label{prop:consistent-commute-multiset}
Let $(\config_1, t_1) \to^{o_1}_\df \dots \to^{o_n}_\df (\config_{n + 1}, t_{n + 1})$ be a consistent trace.
Let $\config_1 \to^{o'_1}_\df \dots \to^{o'_m}_\df \config_{m + 1}$ be another trace.
If we have the \emph{multiset} inclusion $\{ o'_1, \dots, o'_m \} \subseteq \{ o_1, \dots, o_n \}$, then $\config_{m + 1} \to^*_\df \config_{n + 1}$.
\end{theorem}

Although we have omitted much of the technical details of the proof, we have formalized a machine-checkable proof of \Cref{prop:consistent-commute-multiset} in Verus~\cite{verus}, a verification language, and the proof can be found in our GitHub repository~\cite{anonymous-repo}.

Practically, with this bounded confluence theorem, we can do bounded model checking for the confluence property:
we perform symbolic execution to obtain a trace of symbolic configurations
$\config_1(\varlist{x}) \to^{o_1}_\df \dots \to^{o_n}_\df \config_{n +
1}(\varlist{x})$, and then we can solve for a consistent augmentation $(t_1,
\dots, t_{n + 1})$ for this trace using an SMT solver.
If there is a solution, then it follows from \Cref{prop:consistent-commute-multiset} that any other trace (firing operators only in the multiset $\{ o_1, \dots, o_n \}$, counting multiplicity) would still converge to $\config_{n + 1}(\varlist{x})$ in the end.

\begin{algorithm}
    \caption{Full confluence checking algorithm (\Cref{sec:confluence-full}).}
    \label{alg:confluence-check}
    \begin{algorithmic}[1] 
        \Require{
            Inputs:
            \begin{itemize}
                \item $c^1_\df(\varlist{x}^1), \dots, c^n_\df(\varlist{x}^n)$: a list of dataflow cut points;
                \item $P(L, k)$: a permission algebra.
            \end{itemize}
        }
        \Ensure{If \textsc{ConfluenceCheck} succeeds, then the dataflow program is confluent.}
        \Procedure{ConfluenceCheck}{}
            \State $\psi \gets \top$ \Comment{Permission constraints to be accumulated.}
            \For{$i \in \{ 1, \dots, n \}$} \Comment{Initialize permission variables at cut points} \label{line:confluence-check-permission-init}
                \State $t_i \gets \textbf{empty\textunderscore map}()$ \Comment{$t_i: E \times \mathbb{N} \rightharpoonup PV$ (\Cref{def:permission-augmentation})}
                \For{each channel $e$ and $m \in \{ 1, \dots, |\channel(c^i_\df(\varlist{x}^i), e)| \}$}
                    \State $t_i(e, m) \gets \textbf{fresh\textunderscore variable}()$
                \EndFor
                \State $\psi \gets \psi \wedge \permdisjoint(t_i(e, m) \text{ for all } e, m)$ \Comment{\Cref{def:consistent-augmentation}} \label{line:confluence-check-disjoint}
            \EndFor
            \For{$i \in \{ 1, \dots, n \}$} \label{line:confluence-check-permission-main-loop}
                \State $Q \gets \textbf{empty\textunderscore queue}()$
                \State $\textbf{enqueue}(Q, (c^i_\df(\varlist{x}^i), t_i))$
                \While{$Q$ is non-empty}
                    \State $(c_\df(\varlist{x}^i), t) \gets \textbf{dequeue}(Q)$
                    \For{$(c'_\df(\varlist{x}^i), t') \in \textbf{step\textunderscore canonical}(c_\df(\varlist{x}^i), t)$} \Comment{For each symbolic branch} \label{line:confluence-check-queue-for}
                        \Statex
                        \LineComment{Constraints in \Cref{def:consistent-transition}, where $o$ is the fired operator.} \label{line:confluence-check-step-constraints-start}
                        \State $p_{\mathrm{in}} \gets$ sum of $o$'s input permission variables in $t$
                        \State $p_{\mathrm{out}} \gets$ sum of $o$'s output permission variables in $t'$
                        \State $\psi \gets \psi \wedge (p_{\mathrm{out}} \leq p_{\mathrm{in}})$
                        \If{$o$ reads memory location $l \in L$}
                            \State $\psi \gets \psi \wedge (\readperm\ l \leq p_{\mathrm{in}})$
                        \ElsIf{$o$ writes to memory location $l \in L$}
                            \State $\psi \gets \psi \wedge (\writeperm\ l \leq p_{\mathrm{in}})$
                        \EndIf \label{line:confluence-check-step-constraints-end}
                        \Statex

                        \For{$j \in \{ 1, \dots, n \}$}
                            \If{$\exists \sigma: \varlist{x}^j \to V(\varlist{x}^i) .\ c'_\df(\varlist{x}^i) = c^j_\df(\sigma)$} \Comment{Matched to a cut point} \label{line:confluence-check-subsumption}
                                \For{each channel $e$ and $m \in \{ 1, \dots, |\channel(c^j_\df(\varlist{x}^j), e)| \}$}
                                    \State $\psi \gets \psi \wedge (t'(e, m) = t_j(e, m))$ \label{line:confluence-check-cut-point-constraint}
                                \EndFor
                                \State \textbf{continue} Line~\ref{line:confluence-check-queue-for}
                            \EndIf
                        \EndFor
                        \State $\textbf{enqueue}(Q, (c'_\df(\varlist{x}^i), t'))$
                    \EndFor
                \EndWhile
            \EndFor
            \State \textbf{succeed} iff $\psi$ is satisfiable in $P(L, k)$
        \EndProcedure
    \end{algorithmic}
\end{algorithm}

\subsection{Full Confluence Checking using Cut Points}
\label{sec:confluence-full}

To extend the bounded confluence results in \Cref{sec:confluence-bounded} to checking confluence for the entire dataflow program with unbounded traces, we need to finitely describe and check a consistent pattern of permission augmentations for all traces.
For this purpose, we use the cut points selected in \Cref{sec:sim-canonical} and perform symbolic execution from them not only for symbolic values, but also for permissions.

Let $\cutpoint^1, \dots, \cutpoint^n$ be a list of cut points from the cut-simulation check (\Cref{sec:sim-check}), where each $\cutpoint^i = (c^i_\df(\varlist{x}^i), c^i_\seq(\varlist{y}^i)) \mid \varphi^i(\varlist{x}^i, \varlist{y}^i)$ is a pair of symbolic dataflow and \seqlang configurations with an additional correspondence constraint.
During the cut-simulation check in \Cref{sec:sim-check}, we verify that starting from any dataflow cut point $c^i_\df(\varlist{x}^i)$, any symbolic execution branch reaches an instance of another cut point $c^j_\df(\varlist{x}^j)$.
A consequence is that the symbolic branches from all dataflow cut points ``cover'' any complete trace in the canonical schedule (\Cref{def:canonical-schedule}).
In other words, any complete concrete trace in the canonical schedule is the concatenation of instances of symbolic branches.

Thus, if we are able to find a consistent permission augmentation for the trace of each dataflow symbolic branch during the cut-simulation check, then any concrete trace in the canonical schedule would have a consistent permission augmentation, which implies that any other schedule will converge to the canonical schedule, and the confluence of the entire dataflow program follows.

This full confluence check algorithm is procedurally described in \Cref{alg:confluence-check}.
The algorithm is parameterized with a list of dataflow cut points $c^1_\df(\varlist{x}^1), \dots, c^n_\df(\varlist{x}^n)$ and a predetermined permission algebra $P(L, k)$.
The algorithm first generates a fresh permission variable for each value present in the channels of the dataflow cut points.
Then it performs symbolic execution from each cut point, carrying constraints for both symbolic values and their corresponding permission variables, until the symbolic branch is subsumed by another cut point.
In the end, it checks for the satisfiability of the permission constraints.
If they are satisfiable, the confluence of the full program follows.

In more detail, the first loop at Line~\ref{line:confluence-check-permission-init} assigns fresh permission variables to each value present in each cut point, representing an unknown concrete permission token to be solved for.
We also add disjointness constraints at Line~\ref{line:confluence-check-disjoint} to enforce that these permission augmentations are consistent (\Cref{def:consistent-augmentation}) at the corresponding cut point configurations.

Then in the main loop at Line~\ref{line:confluence-check-permission-main-loop}, we symbolically execute from each dataflow cut point, trying to match each symbolic branch against other cut points.
The step function \textbf{step\textunderscore canonical} uses a deterministic canonical schedule (\Cref{def:canonical-schedule}) we constructed from \seqlang execution traces.
For example, in the case of LLVM and RipTide, we use the canonical schedule
described in \Cref{example:canonical-schedule}.
\textbf{step\textunderscore canonical} is also instrumented to take a permission augmentation $t$ for the current configuration $c_\df(\varlist{x}^i)$, and output a modified permission augmentation $t'$ for the stepped configuration $c'_\df(\varlist{x}^i)$, in which fresh variables are generated for all new output values in $c'_\df(\varlist{x}^i)$, and permissions corresponding to input values are removed.
In each step, for each symbolic branch, we also accumulate permission constraints (Lines~\ref{line:confluence-check-step-constraints-start}-\ref{line:confluence-check-step-constraints-end}) to enforce that the augmented transition $(c_\df(\varlist{x}^i), t)$ to $(c'_\df(\varlist{x}^i), t')$ is a consistent transition (\Cref{def:consistent-transition}).

At each symbolic step, we check if a symbolic branch is matched to another cut point $c^j_\df(\varlist{x}^j)$ (Line~\ref{line:confluence-check-subsumption}).
This check is performed by finding a substitution
$\sigma: \varlist{y}^j \to \term(\varlist{x}^i)$ such that the current symbolic branch $c'_\df(\varlist{x}^i) = c^j_\df(\sigma)$.
If the symbolic branch matches the cut point, we add equality constraints (Line~\ref{line:confluence-check-cut-point-constraint}) to enforce that the permissions in $t'$ should converge back to the original permissions $t_j$ at the matched $j$-th cut point.
Note that the substitution $\sigma$ is not later used because we do not need to check correspondence constraints (\Cref{sec:sim-check}).

The main loop at Line~\ref{line:confluence-check-permission-main-loop} should terminate, because if the cut-simulation check (\Cref{sec:sim-canonical}) passes, any branch from any dataflow cut point should eventually match another cut point.

Finally, we check if the accumulated permission constraint $\psi$ is satisfiable in the permission algebra $P(L, k)$ via an SMT query.
If $\psi$ is satisfiable, then any concrete trace (of unbounded length) in the canonical schedule has a consistent augmentation.
Thus, by \Cref{prop:consistent-commute-multiset}, any other possible schedule would eventually converge to the canonical schedule, and the dataflow program is confluent.
If $\psi$ is not satisfiable, then there are different possibilities: 1) the dataflow program is not confluent and may have data races (which require manual inspection to confirm); 2) the choice of $k$ in the permission algebra $P(L, k)$ is not large enough; or 3) the confluence property is not provable using this method.


\subsection{Liveness in the Bounded-Channel Model}
\label{sec:bounded-channel}


In our current semantics for dataflow programs, channels between operators are modeled as \emph{unbounded} queues.
In real CGRA architectures, however, the channels have a fixed buffer size, and an operator may block if one of the output channels is full.
As a result, a dataflow program in the bounded-channel model may have a deadlock that is not possible in the unbounded-channel model.

For example, consider \Cref{fig:deadlock}, where operator $\texttt{A}$ has two dataflow paths to $\texttt{B}$:
one path is full, while the other is empty due to a Steer operator $\texttt{T}$
discarding its inputs. 
This situation puts us in deadlock.
Operator $\texttt{B}$ cannot fire without tokens along both paths, while
operator $\texttt{A}$
cannot fire before more space is made along the full path. 

To verify that this scenario does not occur, we lift our cut-simulation and confluence results in the unbounded-channel model to the bounded-channel model via two observations:
\begin{itemize}
    \item The symbolic branches in the dataflow cut point execution (\Cref{sec:sim-check}) have finitely many intermediate symbolic configurations, so the maximum channel size is bounded by some integer $K$ in the canonical schedule.
    \item The confluence results in \Cref{sec:confluence-bounded} and \Cref{sec:confluence-full} are still valid in the bounded-channel model, which we have formalized and verified in Verus as well.
\end{itemize}
Therefore, if the cut-simulation and the confluence checks pass in the unbounded-channel model, the same results hold in the bounded-channel model if channels have a buffer size of at least $K$.
As a result, in the bounded-channel model, the dataflow program is still equivalent to the input imperative program,
which guarantees the liveness property of the dataflow program that it will always make progress if the sequential imperative program is able to make progress.
Furthermore, since the sequential program is deadlock-free, the dataflow program is also deadlock-free.




\newcommand{\evalnumprograms}{20\xspace}

\begin{figure}[ht!]
    \centering
    \begin{tabular}{lrrrrrl}
        Name & LOC & \#OP & \#PC & Sim. & Conf. & Description \\
        \hline
        \texttt{nn\textunderscore vadd} & 5 & 15 & 244 & 0.02 & 0.17 & Vector addition (RipTide) \\
        \texttt{nn\textunderscore norm} & 10 & 19 & 316 & 0.02 & 0.19 & Neural network normalization (RipTide) \\
        \texttt{nn\textunderscore relu} & 11 & 19 & 844 & 0.03 & 0.45 & Neural network ReLU layer (RipTide) \\
        \texttt{smv} & 12 & 25 & 1103 & 0.05 & 2.27 & Sparse-dense matrix-vector mult. (Pipestitch) \\
        \texttt{dmv} & 10 & 31 & 1234 & 0.06 & 1.64 & Dense-dense matrix-vector mult. (RipTide) \\
        \texttt{Dither} & 19 & 31 & 2553 & 0.13 & 2.29 & Dithering (Pipestitch) \\
        \texttt{SpSlice} & 30 & 31 & 2616 & $\times$ 0.13 & 7.94 & Sparse matrix slicing (Pipestitch) \\
        \texttt{nn\textunderscore fc} & 23 & 36 & 2632 & 0.09 & 3.35 & Neural network fully-connected layer (RipTide) \\
        \texttt{SpMSpVd} & 45 & 40 & 3182 & 0.15 & 8.38 & Sparse-sparse matrix-vector mult. (Pipestitch) \\
        \texttt{bfs} & 48 & 48 & 3634 & 0.15 & 7.68 & Breadth-first search (RipTide) \\
        \texttt{dfs} & 45 & 49 & 3762 & 0.17 & 7.76 & Depth-first search (RipTide) \\
        \texttt{smm} & 22 & 49 & 3967 & 0.16 & 11.2 & Sparse-dense matrix mult. (RipTide) \\
        \texttt{nn\textunderscore pool} & 32 & 52 & 8165 & 0.23 & 7.14 & Neural network pooling layer (RipTide) \\
        \texttt{dmm} & 16 & 56 & 4220 & 0.19 & 6.16 & Dense-dense matrix mult. (RipTide) \\
        \texttt{sconv} & 22 & 56 & 4272 & 0.2 & 10.37 & Sparse convolution (RipTide) \\
        \texttt{sort} & 23 & 57 & 6554 & $\times$ 0.14 & 3.31 & Radix sort (RipTide) \\
        \texttt{SpMSpMd} & 39 & 60 & 7693 & 0.33 & 26.68 & Sparse-sparse matrix mult. (Pipestitch) \\
        \texttt{nn\textunderscore conv} & 42 & 61 & 12154 & 0.2 & 15.32 & Neural network convolution layer (RipTide) \\
        \texttt{dconv} & 24 & 65 & 8134 & 0.25 & 11.51 & Dense convolution (RipTide) \\
        \texttt{fft} & 29 & 70 & 2690 & 0.12 & 5.77 & Fast Fourier transform (RipTide) \\
        \texttt{sha256} & 81 & 181 & 3548 & $\times$ 0.28 & 12.6 & SHA-256 hash
    \end{tabular}
    \caption{Evaluation results on \numevalprograms test cases sorted by the number of dataflow operators. From left to right, the columns are the name of the test case (Name), lines of code in the original C program (LOC), number of dataflow operators (\#OP), number of permission constraints (\#PC), time spent for simulation check in seconds (Sim.), and time spent for confluence check in seconds (Conf.).
    All test cases are originally written in C and then compiled to LLVM via Clang~\cite{clang}.
    Compiler bugs cause the simulation checks to fail in \texttt{SpSlice}, \texttt{sort}, and \texttt{sha256}. 
    In all other cases, the simulation and confluence checks succeed.}
    \Description{Evaluation results.}
    \label{fig:evaluation}
\end{figure}

\section{Implementation and Evaluation}
\label{sec:evaluation}

We have implemented our translation validation technique in the tool \project targeting the RipTide compiler from LLVM programs to dataflow programs.
It is publicly available in a GitHub repository~\cite{anonymous-repo}.
We implemented symbolic executors for (a subset of) LLVM and dataflow programs in Python.
Based on these, we implemented the cut-simulation check described in \Cref{sec:sim-canonical} and the confluence check in \Cref{sec:confluence}.
We also instrumented the RipTide compiler to output hints for constructing the canonical schedule (\Cref{example:canonical-schedule}).

\project uses the Z3 SMT solver~\cite{z3} to discharge the verification conditions it generates.
In the simulation check, \project encodes feasibility of path conditions and validity of the correspondence conditions at each cut point into Z3 queries.
In the confluence check, \project produces a list of permission constraints and checks if these constraints are satisfiable with respect to the finite permission algebra (\Cref{def:permission-algebra}) by encoding each permission variable as a set of Boolean variables and the list of permission constraints as an SMT query.



In order to evaluate \project's capability to certify compilation to real
dataflow programs, we have applied \project to a benchmark of \numevalprograms
programs consisting of the benchmark programs in RipTide~\cite{riptide} and
Pipestitch~\cite{pipestitch}, as well as some programs implementing neural
network inference and an implementation of SHA-256~\cite{sha256-impl}.
\Cref{fig:evaluation} shows some statistics about the benchmark programs and how \project performs.
All tests are performed on a laptop with an Apple M1 processor and 64 GiB of RAM.
The time spent in the RipTide compiler (including hint generation) is less than 0.1 seconds for all of the benchmark programs, so it is omitted from the table.

We now evaluate \project in three aspects: generality of cut-point
heuristics, performance, and size of the trusted computing base.%
\paragraph*{Generality of the heuristics}
Both simulation and confluence checks pass for most of the benchmark programs.
This shows that our heuristics for inferring dataflow cut points work well on
the canonical schedule, and our confluence check is general enough for common
compilation patterns.
In test cases \texttt{sort}, \texttt{SpSlice}, and \texttt{sha256}, the simulation check fails and correctly reveals two compilation bugs (see \Cref{sec:compiler-bugs} for more detail).
After fixing these bugs, the simulation and confluence checks succeed for these two programs.




\paragraph*{Performance}
\project spends most of its time on the confluence check, mainly consisting of
checking satisfiability for permission constraints.
The simulation check takes a relatively short amount of time, because the input LLVM program and the output dataflow program match well on the canonical schedule and thus the verification conditions are simple to solve.
Overall, the verification time is within 15 seconds for most benchmark programs, and we believe it is feasible to enable \project in real production scenarios.

\paragraph*{Trusted Computing Base}
The original RipTide compiler has \textasciitilde 21,000 lines of C++ code, excluding comments. 
The instrumented hint generation in the RipTide compiler is implemented in less than 50 lines of C++.
Compared to the RipTide compiler, the trusted computing base of \project has 3,121 lines of Python with three main components: \textasciitilde 800 lines of dataflow semantics; \textasciitilde 1,000 lines of LLVM semantics; and \textasciitilde 750 lines for simulation and confluence checks.

\begin{figure}
    \centering
    \begin{minipage}[b]{0.6\linewidth}
        \centering
        \includegraphics[width=\linewidth]{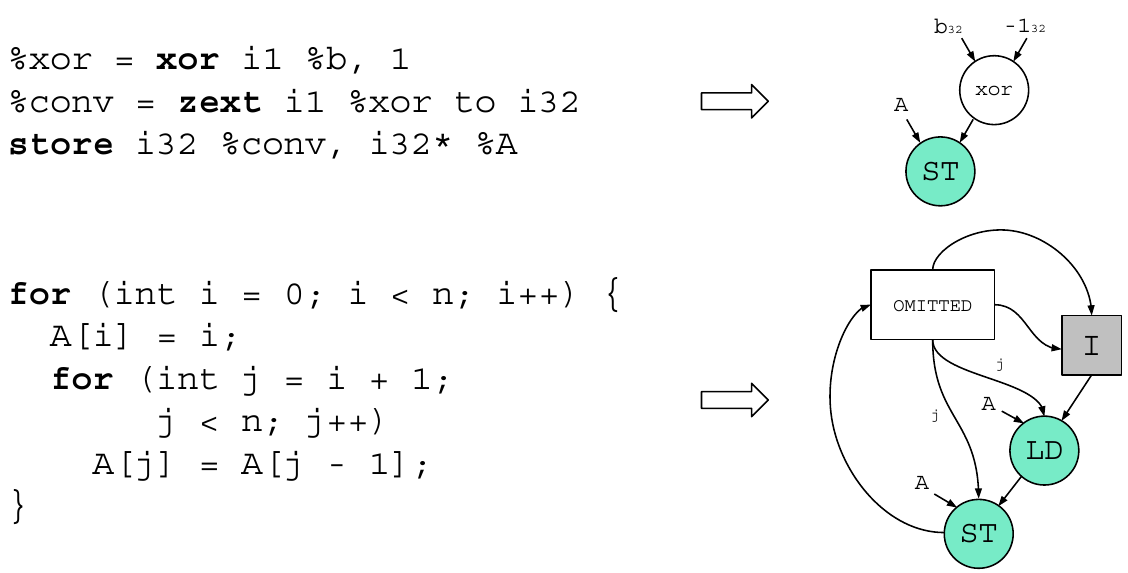}
        \caption{Examples of two compiler bugs.}
        \Description{Examples of two compiler bugs.}
        \label{fig:bugs}
    \end{minipage} \qquad
    \begin{minipage}[b]{0.3\linewidth}
        \centering
        \includegraphics[width=0.65\linewidth]{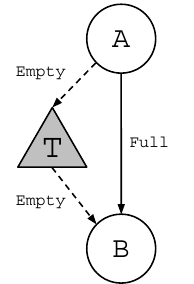}
        \caption{Deadlock example.}
        \Description{Deadlock example.}
        \label{fig:deadlock}
    \end{minipage}
\end{figure}

\subsection{Compiler Bugs}
\label{sec:compiler-bugs}

While testing \project, we found \numbugs bugs in the RipTide compiler, confirmed by RipTide authors.
Four of these bugs were discovered on the test cases \texttt{sort}, \texttt{SpSplice}, \texttt{sha256}, and their variations;
and the remaining four bugs were found on test cases not included in the evaluation.
We now discuss two representative bugs:
\paragraph{Incorrect signed extension}

\project found a compilation bug involving integer signed vs. unsigned extension in the \texttt{sort} test case.
A simplified LLVM program used to reproduce this issue is shown in the top left of \Cref{fig:bugs}, and the output dataflow program is shown in the top right of \Cref{fig:bugs}.
The LLVM code performs a 1-bit xor with a 1-bit constant \texttt{1}, and it zero-extends the result to a 32-bit integer.
On the dataflow side, the dataflow compiler converted all operations into 32-bits (word length of RipTide).
However, during the process, the compiler assumes that all constants should be sign-extended.
As a result, it extends the 1-bit \texttt{1} to a 32-bit \texttt{-1}, which is incorrect in this case.
This difference causes our simulation check to fail, as the final memory states fail to match.

\paragraph{Incorrect memory ordering with potential data races}
Our confluence check catches a compilation bug when compiling from the C program shown in the bottom left of \Cref{fig:bugs} to the dataflow program in the bottom right of \Cref{fig:bugs}.
The C program has a nested loop.
In each iteration of the outer loop, it first sets \texttt{A[i]} = {i}, and then in the inner loop that follows, it copies \texttt{A[i]} into \texttt{A[i + 1]}, \dots, \texttt{A[n - 1]}.
In the output dataflow program, the compiler allows the load and store operators in the inner loop to pipeline using the invariant operator (marked \texttt{I}) to repeatedly send the signal to enable load.
However, this causes a data race, because in the inner loop, a load in the second iteration accesses the same memory location as the store in the first iteration.

The confluence check fails as the generated permission constraints are unsatisfiable.
\project also outputs an unsat core of 57 constraints, a subset of the 420 constraints generated, that causes the unsatisfiability.
In particular, the unsat core includes two constraints:
$\readperm\ \texttt{A}$ is in the input tokens of the inner loop load, and $\writeperm\ \texttt{A}$ is in the input tokens of the inner loop store.
This indicates that the main conflict stems from these two operators.

\section{Limitations and Future Work}
\label{sec:limitations}

Our translation validation technique still has some limitations to address in future work.

\paragraph{Optimized dataflow graphs}

Our technique may fail on more optimized programs.
The RipTide compiler employs some optimizations currently not supported by \project:
\begin{itemize}
    \item Deduplication of operators, which combines arithmetic operators with the same inputs.
    \item Streamification, which uses a single stream operator to generate a loop induction variable (as opposed to using, e.g., carry, comparison, and addition operators).
    \item Array dependency analysis, which removes unnecessary orderings between loads and stores across iterations that access different indices of an array. 
\end{itemize}
The first two optimizations are not supported by our simulation check because we require a one-to-one correspondence between LLVM instructions and dataflow operators for the canonical schedule.
Since these optimizations involve reasoning about the equivalence between dataflow programs, we believe that a modular approach is to use \project to verify the input LLVM against the unoptimized dataflow program, and use another verification pass in future work to soundly optimize the dataflow program.

The third optimization may lead to a dataflow program failing our confluence check.
To solve this, we need finer granularity in permission tokens.
Currently, the \readperm\ and \writeperm\ permissions are for the entire memory region, such as an array \texttt{A}.
To allow memory operations on \texttt{A[i]} and \texttt{A[i + 1]} to run in parallel, where \texttt{i} is the loop induction variable, we need to re-formulate the permission algebra to allow slices of a permission token on a specific range of indices.
This also requires a form of ``dependent'' permission tokens that can depend on the values in the channels.

\paragraph*{Scalability of the confluence check}
\project currently spends most of its time on the confluence check (\Cref{sec:evaluation}).
We believe this is primarily due to a naive encoding of permission constraints in SMT.
For better performance in future work, one could design a custom solver for permission constraints to better utilize the structure of a permission algebra, or use a preprocessing step to simplify permission constraints.




\paragraph{Lack of architecture-specific details}

This work targets the compilation from imperative programs to an abstract version of dataflow programs.
While our technique can ensure the correctness of this pass, we do not verify later passes, such as the procedure to optimally configure the dataflow program on the actual hardware.


\section{Related Work}
\label{sec:related}

As a compiler verification technique, the translation validation approach taken in our work is a less labor-intensive alternative to full compiler verification, such as CompCert~\cite{compcert}.
Translation validation was first proposed by \citet{pnueli-tv}, and there are numerous following works with novel translation validation or program equivalence checking techniques for various languages and compilers~\cite{necula-tv,kundu-tv,data-driven-eq,tate-eq-saturation}.
For LLVM-related works on translation validation, we have Alive~\cite{alive,alive2}, which performs bounded translation validation for internal optimizations in LLVM.
LLVM-MD~\cite{llvm-md,denotational-tv} is another translation validation tool for LLVM optimizations.
LLVM-MD translates the LLVM code into an intermediate representation called synchronous value graphs (SVG)~\cite{denotational-tv}, which is a similar representation to a dataflow program.
The difference is that the semantics of SVG is synchronous (in the sense that values are not buffered in the channels), and they also trust the compilation from LLVM to SVG, whereas in our case, such translation is exactly what we are trying to validate.
Our formulation of cut-simulation in our simulation check is an adaptation from the work on \textsc{Keq}~\cite{keq}, which aims to be a more general program equivalence checker parametric in the operational semantics of input/output languages.

Our use of linear permission tokens is a dynamic variant of the fractional
permissions proposed by Boyland~\cite{fractional-perm} and used in various separation logics~\cite{concurrent-separation-logic} to reason about read-only sharing of references.
Tools based on separation logics such as Iris~\cite{iris} and Viper~\cite{viper} usually require manual annotations for permission passing in the pre-/post-conditions, whereas our permission tokens are 
synthesized fully automatically.
Furthermore, separation logic is not a suitable formalism to directly reason about dataflow programs due to the lack of structured control-flow in a dataflow program and the asynchrony of dataflow operators.


Various models of dataflow programming~\cite{dataflow-programming-survey} have been studied, such as computation graphs~\cite{computation-graphs} and Kahn process networks~\cite{kahn-process-networks}.
An important aspect of these works is determinacy; i.e., these models produce a deterministic result regardless of the schedule of execution.
There are also verification tools for classical dataflow languages such as Lustre~\cite{lustre,lustre-verified-compiler,lustre-smt-verification} and Esterel~\cite{esterel-verification}.
However, dataflow programs used by CGRA compilers often make use of shared
global memory, which requires stronger techniques for guaranteeing determinacy.


Process calculi such as CSP~\cite{csp} and $\pi$-calculus~\cite{pi-calculus} present an alternative
formalism for message-passing concurrency.
To manage the inherent complexity and nondeterminism in process calculi, session types can enforce confluence and deadlock-freedom by construction~\cite{pi-calculus-typed}.
However, similar to separation logic, session types require manual typing annotations that would be difficult to synthesize automatically, and we likely still need some form of fractional permissions to handle shared memory.

Our strategy for analyzing nondeterministic dataflow programs is reminiscent of
verification strategies for distributed
systems~\cite{canonical-sequentialization,inductive-sequentialization,pretend-synchrony} 
which analyze message-passing protocols (e.g., Paxos~\cite{paxos}) via reductions to sequential programs.
However, the dataflow domain requires significantly different techniques: we
verify \emph{equivalence properties} with arbitrary specification programs, 
and use \emph{linear permissions} to enable safe memory accesses in the presence of 
asynchrony.




In the context of term rewriting systems~\cite{term-rewriting}, the standard approach for confluence checking is the use of critical pairs~\cite{knuth-critical-pairs, duran-confluence}, which proves ground confluence in the case when the rewriting system is terminating.
This algorithm is implemented in various systems, such as the Church-Rosser checker in Maude~\cite{maude}.
However, this technique is difficult to apply in our case, since the semantics of a dataflow program is neither terminating nor, in general, ground confluent.


\section{Conclusion}
\label{sec:conclusion}

In this work, we develop a technique for translation validation between
imperative source programs and asynchronous dataflow target programs, with an
implementation for the RipTide CGRA architecture.
Our technique simplifies analysis on dataflow by first verifying that the target
program simulates the source program along a canonical schedule, and then using
linear memory permissions to show that this canonical schedule is general. 
Our verification procedure ensures both functional correctness and liveness for the target dataflow program.

In future work, we aim to support optimizations on dataflow programs (e.g.,
special-purpose operators for pipelining loops), and extending our verification
methodology to concrete instantiations of CGRAs in hardware.

\bibliography{references}


\end{document}